\documentclass{article}

\usepackage[T1]{fontenc}
\usepackage[utf8]{inputenc}
\usepackage{lmodern}
\usepackage{microtype}
\usepackage{algorithm}

\usepackage[total={5.8in, 7.8in}, asymmetric, bindingoffset=0.4in]{geometry}
\usepackage{amsmath}
\usepackage{amssymb}
\usepackage{amsthm}
\usepackage{mathtools}
\usepackage{enumerate}
\usepackage[dvipsnames]{xcolor}
\usepackage{dsfont}

\usepackage{multirow}
\usepackage{empheq}
\usepackage{listings}
\usepackage{color}
\usepackage{subcaption}

\def\eps{\epsilon}

\def\be{\begin{equation}}
\def\ee{\end{equation}}

\def\bea{\begin{align}}
\def\eea{\end{align}}

\def\bea*{\begin{align*}}
\def\eea*{\end{align*}}

\def\cref#1{Figure~\ref{#1}}
\theoremstyle{plain}
\newtheorem{theorem}{Theorem}[section]

\newtheorem{lemma}[theorem]{Lemma}

\newtheorem{proposition}[theorem]{Proposition}
\newtheorem{asm}[theorem]{Assumption}
\theoremstyle{definition} 

\newtheorem{example}{Example}
\newtheorem{remark}[theorem]{Remark}

\counterwithin{equation}{section}
\counterwithin{example}{section}

\DeclareMathOperator{\E}{\mathbb{E}}
\DeclareMathOperator{\F}{\mathbb{F}}

\DeclareMathOperator{\Q}{\mathbb{Q}}
\DeclareMathOperator{\R}{\mathbb{R}}

\DeclareMathOperator{\calF}{\mathcal{F}}

\DeclareMathOperator{\calS}{\mathcal{S}}
\DeclareMathOperator{\calT}{\mathcal{T}}

\DeclareMathOperator{\calX}{\mathcal{X}}

\DeclareMathOperator{\cF}{\mathcal{F}}

\DeclareMathOperator{\cI}{\mathcal{I}}

\DeclareMathOperator{\cL}{\mathcal{L}}

\DeclareMathOperator{\cR}{\mathcal{R}}
\DeclareMathOperator{\cS}{\mathcal{S}}
\DeclareMathOperator{\cT}{\mathcal{T}}

\DeclareMathOperator{\cX}{\mathcal{X}}

\DeclareMathOperator{\cZ}{\mathcal{Z}}

\DeclareMathOperator{\frakB}{\mathfrak{B}}

\DeclareMathOperator{\frakF}{\mathfrak{F}}

\DeclareMathOperator{\frakN}{\mathfrak{N}}

\usepackage{graphicx} 
\usepackage[font=small,labelfont=bf]{caption}
\usepackage[square,comma,numbers]{natbib}

\usepackage{hyperref}


\numberwithin{equation}{section}
\usepackage[symbol]{footmisc}

\usepackage{fancyhdr}
\date{\vspace{-1em}\normalsize{\today}}

\title{Neural Optimal Stopping Boundary\footnote{Authors
declare that they have no conflict-of-interest.}}

\author{A.\ Max Reppen\footnote{Questrom School of Business, Boston University, Boston, MA, 02215, USA, email: {\texttt{amreppen@bu.edu}}.}
\and H. Mete Soner\footnote{Corresponding author.
Department of Operations Research and Financial
Engineering, Princeton University, Princeton, NJ, 08540, USA, email: 
{\tt soner@princeton.edu}. Research of Soner
 was partially supported by the National Science Foundation grant
 DMS 2106462.}
\and Valentin Tissot-Daguette\footnote{Department of Operations Research and Financial
Engineering, Princeton University, Princeton, NJ, 08540, USA, email: 
{\tt v.tissot-daguette@princeton.edu}}}

\date{\today}

\begin{document}
\maketitle
\vspace{5pt}

\abstract{ A  method based on deep
artificial neural networks and empirical risk minimization
is developed to calculate the boundary
separating the stopping and continuation regions
in optimal stopping.  
The algorithm parameterizes the stopping boundary as the graph of a function and introduces relaxed 
stopping rules based on fuzzy boundaries
to facilitate efficient optimization.
Several financial instruments, some in high dimensions,
are analyzed through this method,
demonstrating its effectiveness.
The existence of the stopping boundary 
is also proved under natural
structural assumptions.}
\vspace{10pt}
\vspace{10pt}

\section{Introduction}
\label{intro}

The classical decision problem of optimal
stopping has found amazing range of applications
from  finance, statistics, marketing, phase transitions to engineering.
While over the past decades many efficient methods have been 
developed for its numerical resolution, until recently
essentially all computational approaches in high-dimensions
first approximate the maximal value.
The quantity that is equally 
important  is the associated stopping rule, and 
the state is divided into two regions depending 
on whether at a point it is optimal to stop or to continue.
Although
approximate optimal actions  can either be 
calculated directly, or are always
available from the calculated value function
in a greedy policy,
a characterization and computation
of these regions through
interpretable criterion would
provide an immediate
link between the policy and the state.

In many important
applications, 
continuation and stopping regions are separated
by graphs of functions that can be used to construct interpretable 
optimal stopping decisions, and
our goal is to directly  compute these
 \emph{stopping boundaries}  whose performances
are close to the maximal ones.
In lower dimensions, they
 can be   accurately approximated by methods
based on  nonlinear partial differential equations or dynamic programming.
However, in models with many degrees of freedom,
this numerical problem
is still challenging.

Our approach is based on a 
method proposed
by E, Jentzen \& Han \cite{HanE,HanJentzenE},
which we call
\emph{deep empirical risk minimization} (deep ERM).
In this framework,  the controller is replaced 
by a deep artificial neural  network
and the random system dynamics is simulated 
via Monte-Carlo.
Recent advances in optimization tools 
and computational capabilities enable us to 
optimize this system efficiently,
thus providing an accurate tool
potentially also applicable to other
problems with stopping boundaries, including models with regime change
or other singularities
 in financial economics,  phase transitions in continuum mechanics,
and Stefan type problems in solidification.

Deep ERM applies essentially to all problems that 
can be formulated as dynamic stochastic
optimization, maybe after problem specific modifications. 
This and similar algorithms 
\cite{BHLP,Becker1,Becker2,Becker4,BGTW,BGTWM,GMS,HPBL,RW,SST,War}
have been applied to many 
classical problems from quantitative finance,
financial economics, and  to nonlinear partial
differential equations that  have stochastic
representations. 
Also, the application to optimal stopping
is carried out by Becker, Jentzen \& Cheridito \cite{Becker1,Becker2}
to compute an approximation
of the optimal stopping rule and
hence, the stopping region as well as the maximal value.
They
accurately price many American options
of practical importance in very 
high dimensions with confidence
intervals, showcasing
the power and the flexibility of deep ERM 
in optimal stopping as well.  
We refer the readers to our accompanying papers \cite{RS, RST}
for an
introduction of deep ERM, and to
the  excellent surveys \cite{FMW,HPW,RWs}
and the references therein for more information.

The stopping regions are exactly the
sets of all points at which the value and the pay-off 
functions agree, and they can be derived from
the calculated value or 
computed directly.  However, 
 qualitatively these approximations  may not 
 always have 
accurate  geometric properties,
and lack interpretability.
Thus a direct computation 
of these regions
with their known structures
is desirable,
and it is the main goal
of our approach.  
Indeed, in many 
financial applications they are characterized 
as the epi or hypo-graph of certain functions  in a natural
but problem specific 
coordinate system,
described 
in Assumption \ref{asm:star}.
We first outline
the algorithm and its properties under this assumption,
and  then verify it in Theorem \ref{t.fb}
for all examples that we consider.
In contrast to methods mapping each state point to stopping decisions, our representation of stopping boundaries as graphs provides topological guarantees, even when approximated.

In our method, the graph separating the
continuation and stopping sets is 
approximated by a deep artificial neural network
which then provides a stopping rule 
used to calculate an empirical reward function.
However,  it is not possible to directly use any gradient
method for  optimization
over these ``stop-or-go'' decisions.
Therefore,  although at the high-level our approach 
is deep ERM, it requires one important modification
to overcome this difficulty.  We replace these sharp  rules
based on the hitting times, by stopping probabilities.
In this relaxed formulation, 
when the state is not
close to the boundary,
we still stop or continue the process
as in the non-relaxed case.
But when it  is close to the boundary,
in a region called the \emph{fuzzy} (\emph{stopping}) \emph{boundary}, we stop
with a probability proportional to the distance to the 
boundary.  
This is analogous to mushy regions in solidification 
or phase fields models in continuum mechanics
 \cite{BSS,B-mushy,Mort,osher,Soner,V-mushy}.  
Similarly  \cite{Becker1,Becker2} 
uses stopping probabilities as relaxed control variables
but without any connection to the geometric structures.

We numerically study several 
American and Bermudan options
with this methodology.
The numerical results for the American 
put options in Black \& Scholes, and in Heston models
verify the effectiveness of the method.
The Bermudan  max-call option, studied in high dimensions,
shows once again the power of  deep ERM.  The numerical
results for the look-back options 
is an example of the flexibility of
this approach in handling path-dependent options.
When possible all our computations 
is benchmarked to previous studies,
and shown to be comparable.  Also, confidence intervals and
upper bounds 
 computed in \cite{Becker1,Becker3,Becker2}, provide
computable guarantees.
Additionally, the stopping regions, 
Figures \ref{t.2}, \ref{t.3}, \ref{fig:asymCall},
for the two-dimensional
max-call options
are qualitatively similar
to those obtained in \cite{BD}.

Direct computation of the stopping boundary
has been the object of several other studies as well.
In low dimensions, techniques from
nonlinear partial differential equations
modeling obstacle problems and phase transitions
allow for accurate calculations.
Also, Garcia \cite{Garcia} proposes a method
similar to ours based on a parsimonious parametrization
of the stopping region.
A recent study by 
Ciocan \& Mi\u{s}i\'{c} \cite{CM} proposes
a tree-based method to compute this partition.
Alternatively, deep Galerkin method for the differential
equations is used by
\cite{SS} to compute the value 
function of basket options
at every point and the stopping boundary
is derived from this function.
\cite{WP} studies the classical
Stephan problem of melting ice in low dimensions,
by the same approach.  We also refer the readers
to \cite{BTW} for a review 
of numerical methods based on boundary parametrization
as well as other approaches to optimal stopping.

For American option pricing with many degrees of freedom,
the main alternative computational tool to deep ERM
are the Monte-Carlo based regression methods described in
Glasserman \cite{Gbook}, 
Longstaff \& Schwartz \cite{LS} and Tsitsiklis \& van Roy 
\cite{TvR}.  
The classical book of Detemple \cite{Dbook}
and the recent article of Ludkovski \cite{Lu} 
provide extensive information on simulation based computational techniques
for optimal stopping.
The key difference between these approaches, and the
value computation through deep ERM
is the choice of the set of basis functions that grows rapidly 
with the underlying dimension. 
Deep artificial neural networks
do not require an  \emph{a priori} specification of the basis.
Therefore, they are likely to be more effective 
for problems with many states, as supported
by the experiments. Indeed, \cite{Becker3,Kohler}
exploit this property and replace the 
preset hypothesis class
by an artificial neural network in their
regression.

The power of these new approaches 
are more apparent in high dimensional settings,
and an interesting example
is the American options  with rough volatility models.
These are infinite-dimensional models,
and their numerical analysis
is given in  Bayer \emph{et al.}~\cite{BTW,BQY} and
in Chevalier \emph{at.~al.}~\cite{CPZ}
by alternative methods.
In a different class of problems with many states,
\cite{GMPW} exploits the symmetry of the
underlying problem to build 
an appropriate architecture of the neural networks.
This approach is particularly relevant for mean-field
games which models many agents that are identical.
Additionally, \cite{Becker1,Becker3,Becker2} consider
several models with many variables
including the fractional Brownian
motion.

The paper is organized as follows.  The problem and 
the value function is defined in the next section.
 The algorithm is introduced in Section \ref{s.algo},
 and the financial structure is outlined in Section \ref{s.f}.
 Section \ref{s.architecture} details the network
 architecture and parameters used in our experiments.
 One dimensional put options in the Black \& Scholes model
 is studied in Section \ref{s.putBS}, 
 in the  Heston model 
in Section \ref{s.putH}, Bermudan 
 max-call options are the topic
 of Section \ref{s.maxcall}, 
and the look-back options are the foci of Section \ref{s.look-back}.
Appendix \ref{a.existence} proves Theorem \ref{t.fb},
Appendix \ref{a.re} proves the convergence 
of the reward function $\cR_\eps$
as the width of the mushy region tends to zero.
Appendix \ref{a.algo} provides the details of the algorithm.

\section{Optimal Stopping}
\label{s.stopping}

Let $T>0$ be the finite time horizon, and
$(\Omega, \Q)$ be the probability space
with the filtration $ \F=(\cF_t)_{t \in [0,T]}$.
The state process $X = (X_t)_{t\in [0,T]}$
is an $\F$-adapted, continuous, Markov process
taking values in a Euclidean space $\cX$.
For  a fixed subset $\cT \subset [0,T]$,
$\vartheta = \vartheta(\cT)$ is the set of 
$\F$-stopping times taking values in $\calT$.
Then, the \emph{optimal stopping} problem,
corresponding to a given \emph{reward function} 
$\varphi$,
is
\begin{equation}
\label{eq:OS}
\text{to maximize}
\ \ 
v(\tau):=  \E[ \varphi(\tau,X_{\tau})],
\qquad
\text{over all}\  \tau \, \in \, \vartheta,
\end{equation}
where for some $a>0$, $\varphi(t,\cdot) \in \cL_a$ defined as
\be
\label{e.lca}
\cL_{a}:= \left\{ \phi :\cX \to \R\ :\ \text{continuous and}\
|\phi(x) | \le C[1+|x|^a]\ \
\text{for some}\ C>0 \right\}.
\ee

In financial examples, $a$ is almost
always equal to one.  Also, $\cT=[0,T]$
corresponds to an American option,
while $\cT$ is a finite set for Bermudan 
ones.  However, since for numerical calculations
one has to discretize the time variable, from the onset
we assume that 
$\cT=\{t_0, t_1, \ldots,t_{n-1}, T\}$,
with $n$ being large for the American ones.

We next define  the value function 
which is a central tool in Markov optimal control  \cite{FS}.
For $t \in \cT$ and $x \in \cX$,
let $\vartheta_t$ be the set
of all $[t,T]\cap \cT$-valued stopping times,
and set
$$
v(t,x):= \sup_{\tau \in \vartheta_t}\ v(t,x,\tau)
\quad \text{where}
\quad
v(t,x,\tau):=
\E\left[  \varphi(\tau,X_\tau) \mid X_t=x\right].
$$

Finally, we introduce our notation  $\R_+=(0,\infty)$,
and 
$$
\cT^\circ:= \cT \setminus\{T\}
=\{t_0, t_1, \ldots,t_{n-1}\}.
$$

\subsection{Stopping region}
\label{ss.cS}

The following subset of the state space is the 
\emph{stopping region} at time $t\in \cT$:
\[
  \cS_t := \left\{ x \in \cX \ :\ \varphi(t,x) = v(t,x) \right\}.
\]
As  $v(T,x)=\varphi(T,x)$ for every $x\in \cX$, we have $\cS_T=\cX$.

Because the stopping regions are closed (hence Borel) and
$\cS_T=\cX$, the hitting time 
\be
\label{e.taustar}
\tau^* := \min \left\{ u \in \cT\     :\
X_u \in \cS_u \right\} 
\ee
is a well-defined stopping time and belongs to $\vartheta$.
Moreover,   $\tau^*$
is optimal, or equivalently,  with $v(\tau)$  as in \eqref{eq:OS},
$$
v(\tau^*)=\sup_{\tau \in \vartheta} v(\tau).
$$

\subsection{Stopping boundary: graph representation}
\label{ss.free}

Our goal is to represent the stopping region by a boundary given by the graph of a function.
Although even in cases where this is not possible, such as for straddle options, the stopping region can be naturally represented as unions and intersections of this structure and treated similarly.

To allow the flexbility of a graph representation, we construct the graph boundaries in a possibly different coordinate system than the natural one of the state space.
Examples of this are illustrated in \cref{fig:homeomorphism}, showing possible boundaries and coordinates for a  max-call option on two symmetric assets; see Examples \ref{ex.maxcall.coordinates}, \ref{ex.maxcall}.
Theorem~\ref{t.fb} shows that for a large class of options, there is a natural coordinate choice, and the many examples in Section~\ref{s.f} show that the choice is often easy also in other cases.

\begin{figure}[t]
\caption{Stopping region of a max-call option on two symmetric assets. The upper connected component becomes an epigraph through $(\alpha,\Xi)$.}
\vspace{-2mm}
\begin{subfigure}[b]{0.49\textwidth}
    \centering
    \caption{Stopping region in $\calX$}
    \includegraphics[height=1.75in,width=2.0in]{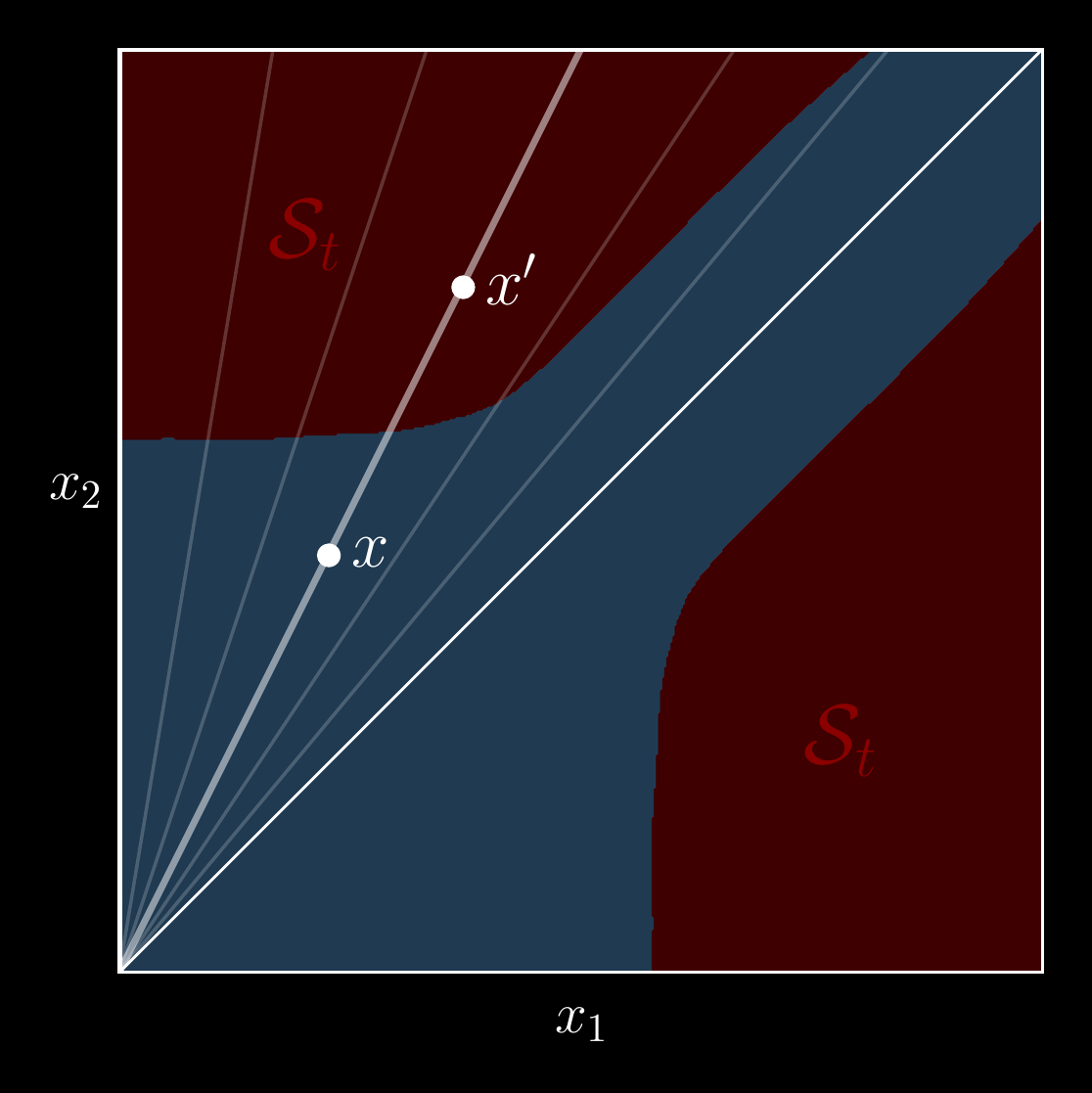}
    \label{fig:X}
\end{subfigure}
\begin{subfigure}[b]{0.49\textwidth}
    \centering
    \caption{Stopping region through $(\alpha,\Xi)$}
    \includegraphics[height=1.75in,width=2.0in]{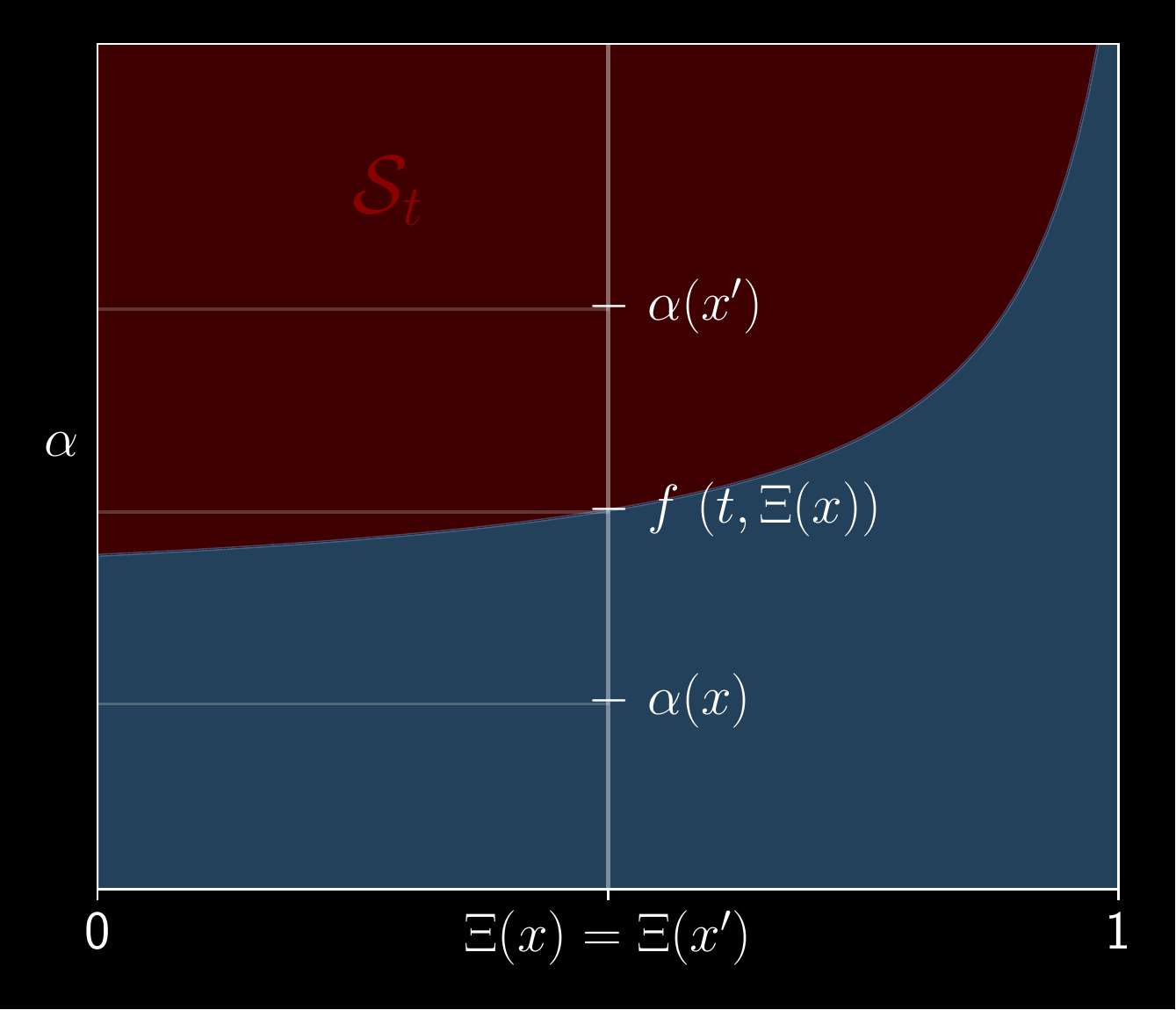}
    \label{fig:homeo}
\end{subfigure}
\label{fig:homeomorphism}
\end{figure}

Let us first consider  three simple examples:

\begin{example}[1D American put]
\label{ex.oned.coordinates}
It is well known that there exists a boundary $f^*$ as a function of time such that
\[
  \cS_t := \left\{ x \in \cX \ :\ x \leq f^*(t) \right\}.
\]
In this case, stopping is optimal in the state $(t,x)$ if $x \leq f^*(t)$, and the indifference points satisfy $x = f^*(t)$.
In other words, the stopping region is the hypograph of $f^*$ and $f^*$ is the stopping boundary.
\end{example}

\begin{example}[Straddle]
\label{ex.straddle.coordinates}
A straddle option involves the both the option of selling and buying the underlying asset. The payoff is therefore given by $\varphi(t,s) = e^{-rt}|s-K|^+$. 
 As the reward appreciates when the stock moves away from the strike, the stopping region has two boundaries: one lower for selling and one upper for buying, cf.\ Figure~\ref{fig:straddleBS}.
This stopping region is naturally the union of two stopping regions delimited by function graphs:
\[
  \cS_t = \left\{ x \in \cX \ :\ x \geq f^U(t) \text{ or } x \leq f^L(t) \right\} = \left\{ x \in \cX \ :\ x \geq f^U(t)\right\} \cup \left\{ x \in \cX \ :\ x \leq f^L(t) \right\},
\]
where $f^L, f^U: \calT^{\circ}\to [0,\infty]$ are the lower and upper boundaries, respectively.
Although this does not directly fit our main assumption, which assumes a single boundary, this more general case is a natural generalization, cf.\ Remark~\ref{rem:union_intersection}.
\end{example}

In both of the preceding examples, the boundary is described by graphs in the problem's natural coordinates.
Even if this is not the case, it is often possible to construct a coordinate system in which the graph representation holds.

\begin{example}[Max-call option]
\label{ex.maxcall.coordinates}
In the financial problem of an American 
max-call option with two stocks
and a given strike $K$, we have
$\varphi(t,x)= (\max\{ x_1,x_2\} - K)^+$
 for $x=(x_1,x_2) \in \cX$.  
With Markovian dynamics, the state space is $\cX =\R_+^2$.
The stopping region is illustrated in Figure~\ref{fig:homeomorphism}, and the coorindates $(\alpha, \Xi)$ from Theorem~\ref{t.fb} are $\alpha = \max\{ x_1, x_2\}$ and $\Xi = x / \max\{ x_1, x_2\}$.
\qed
\end{example}

With these examples in mind, we formulate our main assumption:
that a stopping boundary exists as a function graph (or union thereof), at least in some coordinates.
In Theorem \ref{t.fb} below, we verify this assumption for a large class of options,
under mild natural structural conditions.
More examples are provided in Section~\ref{s.f}.

Although the description of the following assumption is technical, 
its interpretation is quite natural.
In words, we assume that there exists a stopping boundary 
that can be characterized as the graph of some function 
$f^*$ in the appropriate coordinates $(\alpha, \Xi)$.
In the applications, the choice of these functions
is natural as discussed in
Example \ref{ex.maxcall.coordinates} above,
and in Section~\ref{s.f}. 

\begin{asm}{\rm{\bf{(Existence of a stopping boundary)}}}
\label{asm:star}
There exist measurable functions
$$
\alpha: \cX \to \R_+,
\quad
\Xi: \cX \to \Xi(\cX) ,
\quad
 f^* : \cT^\circ \times \ \Xi(\cX) \to [0,\infty],
 $$
 and $\eta \in \{-1,+1\}$,
 so that the map $x \in \cX \mapsto (\alpha(x), \Xi(x))
 \in \R_+ \times \ \Xi(\cX)$ is a homeomorphism
 {\rm{(}}i.e., one-to-one, onto,  continuous with a continuous inverse{\rm{)}}
 and for every $t \in \cT^\circ$,
\be
\label{e.st}
\cS_t = \{  x \in  \cX \ :\
\eta \left(f^*(t,\Xi(x))-\alpha(x)\right) \le 0  \}.
\ee
\end{asm} 

The parameter $\eta$ determines
whether the stopping region is an 
epigraph or a hypograph of $f^*$.
In the financial applications,
typically $\eta=1$ corresponds to a call
type option and the stopping region is an epigraph.
While $\eta=-1$ is related to put options.  

The surjectivity of $\alpha$ ensures
that this characterization is non-trivial,
and it additionally implies that for every $c>0$, the point corresponding
to the pair $(f^*(t,\Xi(x))+\eta c, \Xi(x))$ is in
the stopping region $\cS_t$, and respectively, to
the pair $(f^*(t,\Xi(x))-\eta c, \Xi(x))$ is not in $\cS_t$.  Therefore,
the graph of $f^*$ in the coordinates
$(\alpha(x),\Xi(x))$ given by $$\{x\in \cX : \alpha(x)=f^*(t,\Xi(x))\},$$
is included in the boundary of $\cS_t$.
In fact, in all examples 
it separates the stopping and continuation regions.  
Thus, we call $f^*$ the optimal \emph{stopping boundary}.

\begin{example}[Example~\ref{ex.maxcall.coordinates} continued]
\label{ex.maxcall}
In Theorem \ref{t.fb}, we show that
 there exists a stopping boundary $f^*$
such that the triplet
$(\alpha, \Xi, f^*)$ 
with  
$$
\alpha(x)= \max\{ x_1,x_2\}, 
\qquad \Xi(x)=\frac{x}{\alpha(x)},
\qquad \eta=1,
$$
satisfies the above assumption.
Figures \ref{t.2}, \ref{t.3}, \ref{fig:asymCall} below provide
examples of these regions. 
In our numerical approach,
the latent variable
$\Xi(x)$ is the input
to the network approximating the stopping boundary.
As $\Xi(x) = \Xi(\gamma x)$
for all $x\in \R_+^2,$  $\gamma > 0$,
every ray emanating from the origin is
mapped to a single point by the map $\Xi$,
and the threshold $f^*(t,\Xi(x))$
separates the stopping and continuation
regions on this ray.
\qed
\end{example}

\begin{remark}\label{rem:union_intersection}
  In the theoretical analysis, we assume that there is a single graph that separates the stopping and continuation region.

  More generally, like in the case of Example~\ref{ex.straddle.coordinates}, there may be multiple boundaries or other structures.
  The algorithm we present naturally extends to stopping regions that are unions or intersections of regions delimited by a single graph.
  This is achieved in the implementation by `or' and `and' operations for the respective set operations.

  We thus emphasize that the algorithm in Section~\ref{s.algo} has considerably flexibility.
  \qed
\end{remark}

\subsection{Value function structure}
To obtain some regularity of the value function
and the stopping regions,
we make the following mild assumption 
satisfied in all applications. Recall
that  $\cL_a$ is defined in \eqref{e.lca}.
\begin{asm}{\rm{\bf{(Growth and continuity)}}}
\label{a.growth}
There exists $a>0$,
such that for every $t \in \cT$,
$\varphi(t,\cdot) \in \cL_a$. 
Moreover, for every $\phi \in \cL_a$, $ t < \bar{t} \in \cT$,
$v_\phi (t,\cdot,\bar{t}) \in \cL_a$, where
$$
v_\phi (t,x,\bar{t}):=
\E[\phi(X_{\bar{t}}) \mid X_{t}=x],
\qquad x \in \cX.
$$ 
\end{asm}
\noindent
The above assumed continuity essentially follows form 
the continuous dependence of the distribution
of the state 
on its initial value, and it is always satisfied
in all our applications.

The following are direct consequences.

\begin{proposition}
\label{p.vlsc} Under the  growth and continuity
Assumption \ref{a.growth},
for each $ t\in \cT$,
the value function $v(t,\cdot) \in \cL_a$,
and the stopping region $\cS_t$ is a relatively closed subset of $\cX$.
\end{proposition}
\begin{proof}
To prove the first statement,
we use backward induction in time
$t_0<t_1<\ldots<t_n=T$. 
First observe that by Assumption \ref{a.growth}, 
$v(T,\cdot)=\varphi(T,\cdot) \in \cL_a$.
For the induction hypothesis,
suppose that $v(t_{k+1},\cdot) \in \cL_a$
for some $t_{k+1}\in \cT$. 
By dynamic 
programming,
$$
v(t_k,x)= \max\left\{ \varphi(t_k,x),\ c(t_k,x) \right\},
\qquad\text{where}
\qquad
c(t_k,x):= \E[v(t_{k+1},X_{t_{k+1}}) \mid X_{t}=x].
$$
As $v(t_{k+1},\cdot)$ is assumed to be in $\cL_{a}$, 
 Assumption \ref{a.growth}
implies that
$c(t_k,x)$
is also in $\cL_a$.
Since the space $\cL_a$ is closed
under maximization, the above equation
implies that $v(t_k,\cdot) \in  \cL_a$.
Hence, $v(t,\cdot) \in \cL_a$ for every $t\in \cT$, and in particular,
it is continuous.   

The stopping region $\cS_t$ is the zero level set 
of  $v(t,\cdot)-\varphi(t,\cdot)$ which is shown to be continuous.
Hence, it is  relatively closed in $\cX$.
\end{proof}

As an immediate consequence of the above assumption,
we can restrict the maximization in \eqref{eq:OS} to 
stopping times given by stopping boundaries.  For future
reference, we record this fact.  Let $\frakF$ be the set of
all measurable functions $f :\cT^\circ \times \ \Xi(\cX) \to [0,\infty]$.
For $f\in \frakF$, the corresponding
stopping time is given by,
$$
\tau_f:=  \min \left\{ u \in \cT\ :\
\eta\left( f(u,\Xi(X_u))-\alpha(X_u)\right) \le 0 \ \ \text{or}\ \ u=T\right\}
\in \vartheta.
$$

\begin{lemma}
\label{l.restrict} Under the growth and continuity Assumption \ref{a.growth},
and the existence of a stopping boundary Assumption
\ref{asm:star}, 
$$
\sup_{\tau \in \vartheta} v(\tau) = \sup_{f \in \frakF} v(\tau_f).
$$
\end{lemma}
\begin{proof}
As $\tau_f \in \vartheta$,
we have $\sup_{\tau \in \vartheta} v(\tau) \ge
v(\tau_f)$  for all  $f \in \frakF$.
Let $f^*$ be as in Assumption \ref{asm:star}.
Then, for $u \in \cT$, we have $X_u \in \cS_u$ if and only if,
either $\eta\left(f^*(u,\Xi(X_u))-\alpha(X_u) \right) \le 0$, or $u=T$.
Therefore, $\tau_{f^*}$ is equal to the 
optimal hitting time $\tau^*$ defined in \eqref{e.taustar}.
These imply that
$\sup_{\tau \in \vartheta} v(\tau) \ge \sup_{f \in \frakF} v(\tau_f)
\ge v(\tau_{f^*}) = v(\tau^*)
= \sup_{\tau \in \vartheta} v(\tau)$.
\end{proof}

\section{The Algorithm}
\label{s.algo}
Our goal is to construct an efficient
algorithm for the calculation
of the interpretable
stopping decisions based on boundaries
approximated by artificial neural networks.  
We train the networks by the \emph{deep
empirical risk minimization} algorithm proposed 
by Weinan E, Jiequn Han, and Arnulf Jentzen 
\cite{HanE,HanJentzenE}.  

This method 
approximates the expected value of a reward 
function by Monte-Carlo simulations and
optimizes it by stochastic gradient ascent.
For the problem of optimal stopping,
the natural choice for the reward function  is
$\varphi(\tau_{\theta},X_{\tau_{\theta}})$. However, the map $\theta \mapsto \tau_\theta$
is piece-wise constant, making
optimization by gradient ascent impossible.
Therefore, we replace the hitting times
by stopping probabilities based on 
fuzzy boundaries that are also used in  some
problems of solidification  \cite{B-mushy,Mort,Soner,V-mushy}. 

\subsection{Fuzzy Stopping Boundary}
\label{ss.fuzzy}
Let $\eta\in \{-1,1\}$ be as in Assumption \ref{asm:star}.
For  $f \in \frakF$, set
$$
d(t,x;f):= \eta \left(f(t,\Xi(x))-\alpha(x)\right), \qquad
t \in \cT^\circ, \ x \in \cX.
$$
In the original problem, once a boundary $f$ is chosen,
we stop
the process only when $d(t,X_t;f)\le0$.
As training is not possible
with this sharp ``stop-or-continue'' rule,
we introduce the fuzzy boundaries.  Namely, we fix a 
\emph{tuning parameter} $\eps>0$, and 
at time $t \in \cT^\circ$,
define 
$$
\frakB_t^\eps(f) :=\{ x\in \cX \ :\ |d(t,x;f)|<\eps\}
$$ 
be the 
\emph{fuzzy region}.
\begin{figure}[H]
\begin{center}
\includegraphics[height=1.8in,width=2.8in]{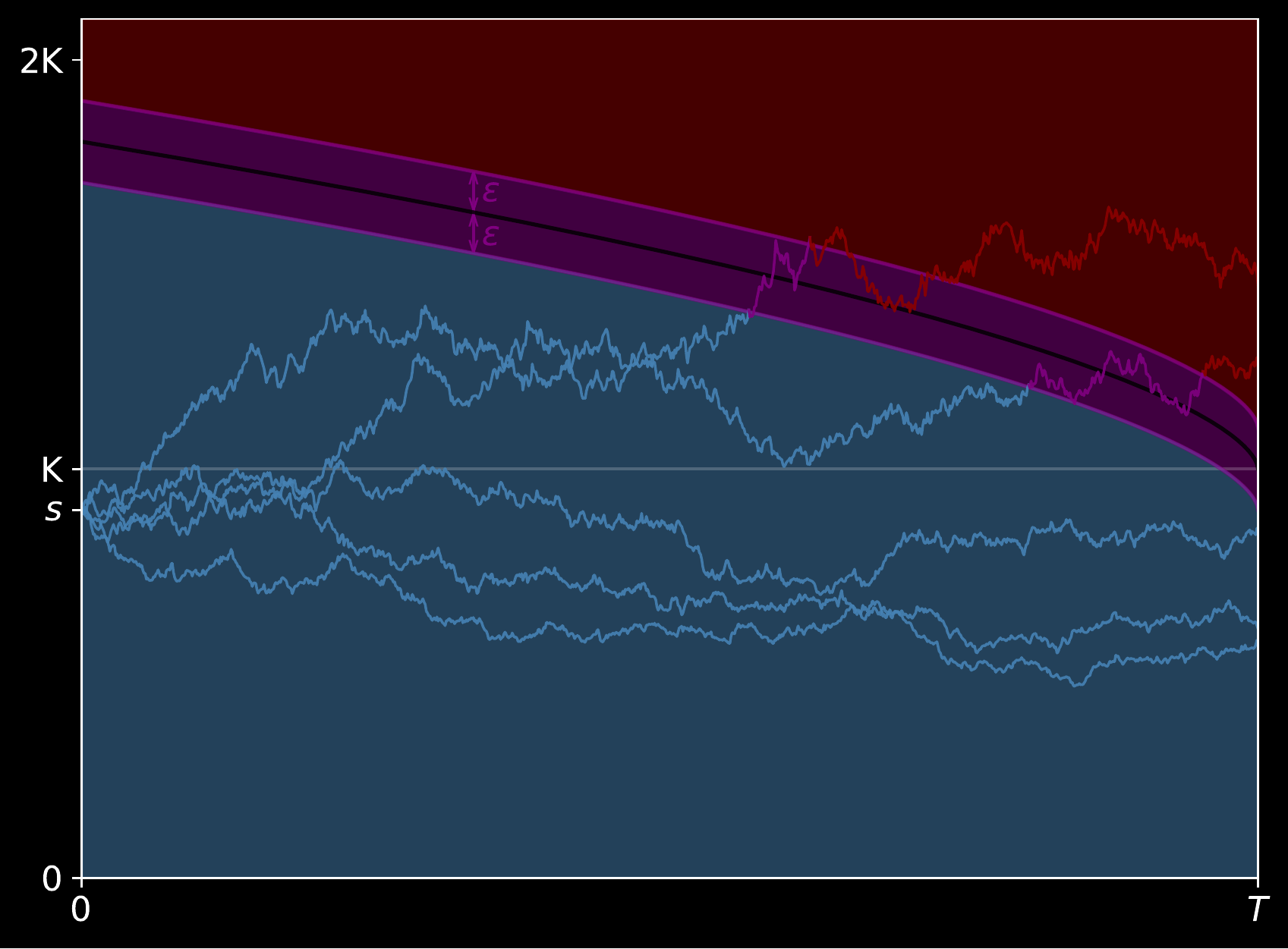}
\caption{
Fuzzy region of an American call (in purple) 
connecting the stopping region (in red) 
to the continuation region (in blue). }

\label{fig:fuzzy}
\end{center}
\end{figure}
\vspace{-0.5em}
\noindent
An illustration of the fuzzy region is given in \cref{fig:fuzzy} for an  
American call option.  If the state is not in
$\frakB_t^\eps(f)$,
we continue or stop, as before, with probability one.  
But if we are in the fuzzy region,
we stop with  probability $p_t(X_t,f)$, where
$$
p_t(x,f):= \left(\frac{\eps - d(t,x,f)}{2\eps}\right)^+\wedge 1,
\qquad   t \in \cT^\circ,\  f \in \frakF, \ x \in \cX.
$$
At maturity, we stop at all states, or equivalently,
we set $p_T(\cdot,f)\equiv 1$.
Precisely, $p_t(X_t,f)$
is the probability of stopping  at time $t \in \cT$,
conditioned on the event that the process is not stopped
prior to $t$.  The relaxed control problem is then defined
using the reward function given by
\begin{equation}
\label{e.ell}
\cR_\eps(X,f):=
\sum_{t\in \cT}\ p_t(X_t,f)\ b_t(X,f)\ \varphi(t,X_t),
\end{equation}
where  $b_t(X,f)$ is the probability
of not stopping strictly before time $t$
and it  is 
obtained as the solution of the following difference equations,
$$
b_{t+1}(X,f)
=b_{t}(X,f) (1-p_{t}(X_t,f)), \qquad t \in \cT^\circ,
$$
with $b_0(X,f)=1$.
In other words,
$b_t(X,f)$ is the unused 
\emph{stopping-budget} remaining.
The process $b_t(X,f)$ is non-increasing in
time and takes values in $[0,1]$.  
Connection
between $\E[\cR_\eps(X,f)]$ and 
$v(\tau_f)$ is 
discussed in Appendix \ref{a.re}.

We refer to \cite{RST} for a different exposition of this relaxation.

\subsection{Deep ERM}
\label{ss.derm}

In this method, we 
approximate the stopping boundary
by a set of deep neural networks
that we abstractly parameterize
by $\theta$ taking values in a finite-dimensional space $\Theta$
as follows:
$$
\frakN=\left\{\ g(\cdot; \theta) :\cT^\circ \times \ \Xi(\cX) \to \R_+\ :
\ \theta \in \Theta\ \right\} \subset \frakF.
$$
Let $\tau_{\theta}=\tau_{g(\cdot;\theta)}$
be the hitting time corresponding to 
the stopping boundary $g(\cdot;\theta)$.
The following is the  pseudocode of the deep
empirical risk minimization based on
the  fuzzy stopping boundary.
\vspace{10pt}

\noindent
\begin{enumerate}
\item \emph{Initialize}  $\theta \in \Theta$.
\item \emph{Simulate}
independent state trajectories,  $\{X^{(1)}, \ldots, X^{(B)}\}$,
where $B$ is the \emph{batch size}. 
\item Compute the \emph{empirical reward}  function:
$$
R_\eps(\theta):=\frac1B\ \sum_{i=1}^{B}\ \cR_\eps(X^{(i)},g(\cdot;\theta)),
$$
where $\cR_\eps(X,f)$ is as in \eqref{e.ell}. A 
modification $\cR_{\eps,\lambda}$, given in \eqref{e.lambda},
allows for \emph{importance sampling} discussed in Section
\ref{ss.important} below.
\item \emph{Optimize and update by stochastic gradient ascent}
with the  learning rate process $\zeta$:
$$
\theta \, \leftarrow \,
 \theta + \zeta  \nabla_\theta R_\eps(\theta).
 $$

\item \emph{Stop} after  $M$ \emph{number of iterations}.
\item \emph{Compute the initial price}
given by the training network, using $J$ many
Monte-Carlo simulations with a sharp boundary.
\end{enumerate}
\vspace{10pt}

A formal discussion of the algorithm is given in the next subsection.  
In  Appendix \ref{a.algo}, we also provide more details
of the above algorithm for the benefit
of readers interested in the implementation. 
\cite{Garcia} proposes a similar
approach by assuming that the stopping region 
has a finite dimensional representation
without a formal justification.
Also,
Chapter 8.2 of \cite{Gbook} discusses a general
approximation for the value calculation by parametrizing
the stopping times. In our approach 
we approximate the stopping boundary directly.

\subsection{Discussion}
\label{ss.discuss}

The general theory of stochastic gradient ascent  \cite{SGD} implies that
the above algorithm constructs an approximation of 
$$
\theta_\eps^* \in \arg\max_{\theta \in \Theta}\ 
\E\left[\cR_\eps(X,g(\cdot;\theta))\right].
$$
In the context of optimal
stopping, \cite{bel}
also provides rates of convergence.
Moreover, in Lemma \ref{l.conv} below, we
show that for each $\theta \in \Theta$,
$$
\lim_{\eps \downarrow 0}\  \E[\cR_\eps(X,g(\cdot;\theta))]=
\E[\varphi(\tau_\theta,X_{\tau_\theta})]= v(\tau_\theta).
$$
Therefore,
$$
\E\left[\cR_\eps(X,g(\cdot;\theta_\eps^*))\right]
= \sup_{\theta \in \Theta}\ 
\E\left[\cR_\eps(X,g(\cdot;\theta))\right]
\approx 
 \sup_{\theta \in \Theta}\ v(\tau_\theta).
 $$
Also, by the celebrated universal approximation theorem \cite{Cybenko,Hornik},
continuous stopping boundaries are well approximated 
by the hypothesis class $\frakN$, and in all our applications
the optimal stopping boundaries are continuous.
Thus, in view of this
approximation capability of $\frakN$ and Lemma~\ref{l.restrict},
we formally have 
$$
\E\left[\cR_\eps(X,g(\cdot;\theta_\eps^*))\right]
= \sup_{\theta \in \Theta}\ 
\E\left[\cR_\eps(X,g(\cdot;\theta))\right]
\approx \sup_{\theta \in \Theta}\ v(\tau_\theta) \approx
\sup_{f \in \frakF}\ v(\tau_f) 
=\sup_{\tau \in \vartheta}\ v(\tau). 
$$

Hence, the above algorithm
constructs an artificial neural network
with an  asymptotic performance formally close to the 
optimal value, justifying
the algorithm.
However, there are several
technical issues in front a rigorous statement
and we postpone this convergence analysis
to another manuscript \cite{SonerTissot}.

\section{Financial Examples}
\label{s.f}
We numerically study several
American or Bermudan options to illustrate
and also assess the proposed algorithm.
In these applications, $\Q$ is a risk-neutral measure and the Markov state process 
$X=(S,Y,Z)$ has three components: $S$ is
the prices of the underlying stocks
possibly including the past values to make it Markov,
a factor process $Y$ used in models like Heston,
and a functional $Z$ of the
stock process, making
the pay-off of a path-dependent option a function of
the current value of $X$.  Typical examples
of $Z$ are the running maximum or minimum in look-back options,
and  the average stock price for the Asian options.
Depending on the model and the option, $Y$, $Z$, or neither may
be present. 

In our examples,
the pay-off function
always has the following form:
\begin{equation}
\label{eq:payoff}
\varphi(t,x) 
= e^{-rt} \left( \eta  (\alpha(s,z) -\beta(s,z)-K)\right)^+, 
\qquad t \in \cT, \ x=(s,y,z) \in \cX, 
\end{equation}
where $\alpha>0, \beta \ge 0$
are  \emph{positively homogenous of degree one},
and the parameter $\eta\in\{-1,+1\}$
determines the option type:
$\eta=1$ is a call and $\eta=-1$  is a put.
The choice of $\alpha$ and $\beta$ may not be 
unique, and we would choose $\beta\equiv0$ 
when possible. It is clear that $\varphi$ 
grows linearly.  Additionally, in all our examples,
the stock price distribution depends
smoothly on the initial data.  Hence, the growth
and continuity Assumption \ref{a.growth} holds with $a=1$.

The following stopping boundary result is proven in the Appendix.
For max-call options a similar result is proved
in Proposition A.5 in \cite{BD}, and for look-back options in \cite{DaiKwok}.
\begin{theorem}
\label{t.fb} Suppose that in \eqref{eq:payoff},
 $\alpha : \calX \to \R_+$,  $\beta: \calX \to [0,\infty)$
are positively homogenous of degree one, and that
 $X$ scales 
linearly in the initial values $S_t=s, Z_t=z$, i.e., for 
$x=(s,y,z)$, $\gamma>0$,
and a bounded, continuous, integrable function $\phi$,
the following holds:
$$
\E\left[ \phi(u,X_u) \mid  X_t =(\gamma s, y, \gamma z)\right]
= \E\left[ \phi(u,\gamma X_u) \mid  X_t =(s, y,z)\right].
$$
Further assume that the growth and continuity 
Assumption \ref{a.growth} holds, and 
the price of the European option is strictly positive
for all $t<T$.  Then, the existence of a stopping boundary Assumption
\ref{asm:star} is satisfied with  $\alpha $
as in \eqref{eq:payoff}, 
\be
\label{e.xi}
\Xi(s,y,z)= \left(\frac{s}{\alpha(s,z)}\,
,\, y\, ,\, \frac{z}{\alpha(s,z)}\right),
\ee
and
\be
\label{e.fstar}
f^*(t,\xi):=
\left\{
\begin{array}{ll}
\sup\left\{ \alpha(s,z)\ :\ x\in \cS_t\ \ \text{s.t.}\ \ \Xi(x)=\xi \right\}, \quad
 & \text{if}\ \ \eta=- 1,\\
\, \inf\left\{ \alpha(s,z)\ :\ x\in \cS_t\ \ \text{s.t.}\ \ \Xi(x)=\xi \right\}, \quad
 & \text{if}\ \ \eta=1,
 \end{array}
 \right.
\ee
where we use the convention that the sup over the empty set is zero and
inf is infinity.
\end{theorem}

The positivity of the European price
ensures that on stopping regions $\varphi>0$.


\subsection{Importance sampling}
\label{ss.important}

To effectively train the network  $g(\cdot;\theta)$
by deep ERM, 
it is  crucial  that the simulated paths
$X_t$ hit the fuzzy  boundary $\frakB_t^\eps(g(\cdot;\theta))$ 
 frequently 
enough for each exercise date. 
We achieve this by importance sampling
that we now outline.

In all our examples,
$S$ evolves according to a stochastic differential equation, 
\begin{align*}
 dS_t = S_t\left(\mu_t dt + \sigma_t dW_t\right),
\quad s_0 \in \R^d,
\end{align*}
where the processes $\mu \in \R^d$,
$\sigma \in \R^{d\times d}$ 
may depend on $t,S_t$ and other exogenous factors. 
For $\lambda=(\lambda_1,\dots,\lambda_d)  \in \R^d$, let
$\Q_\lambda$ be the measure obtained by
 the Girsanov transformation so that
 $W_t^{\lambda} :=  W_t +\lambda t $ is a 
 $\Q_{\lambda}$ Brownian motion.
 Hence, under $\Q_\lambda$,
$$
dS_t =  S_t\left([\mu_t -   \sigma_t \lambda ] dt 
+ \sigma_t dW^{\lambda}_t\right).
$$
We then adjust the hitting
frequency of the simulated trajectories,  
by choosing $\lambda$ appropriately,
thus allowing for more efficient training.
Although in all our experiments we have employed constant $\lambda$,
one could also use time dependent
ones as well.

However, as the probability measure is modified,
we need to account for it in the reward function.
This is achieved by the Radon-Nikodym derivative
given by,
$$
\cZ^{\lambda}_t (S):=
\frac{d\Q}{d\Q_{\lambda}} \big |_{\calF_t} 
= \exp( \lambda \cdot W_t^{\lambda}  -\frac12 
|  \lambda|^2 t).
$$
Since $W^\lambda$ can be expressed in terms of $S$, 
the process $\cZ^\lambda$ 
can be viewed as a function of the state process $S$.
We then modify the reward function given in \eqref{e.ell}
by,
\begin{equation}
\label{e.lambda}
\cR_{\eps,\lambda}(X,f):= 
\sum_{t\in \cT}\ \cZ^{\lambda}_t(S)\ p_t(X_t,f)\ b_t(X,f)\ \varphi(t,X_t).
\end{equation}

\subsection{Network Architecture and Parameters}
\label{s.architecture}
 
All our numerical  experiments have been carried out  
with Tensorflow 2.7  on a 2021 Macbook pro with 64GB unified memory 
and Apple M1 Max chip. The code is implemented in 
Python and run on CPU (10-core) only.  
Throughout the examples, we use the same neural 
network architecture and training parameters,  given  below:
\vspace{5pt}

\noindent
{\bf{1.}} We use only one deep neural network
that takes both the time and 
the state vector as input.  Then, 
the trained boundary is a function of the 
continuous time variable, delivering smoother 
dependence on time.  This is in contrast to
majority of previous studies that
use one deep neural network for each time 
in the discrete set $\cT$.
\vspace{5pt}

\noindent
{\bf{2.}} 
The single feedforward neural network that 
we employ
consists of $2$ hidden layers
 with leaky rectified linear unit (Leaky ReLU) activation function. 
 Each layer has $20+\bar{d}$ many nodes where $\bar{d}$ is the dimension of the latent space $\Xi(\calX)$.  
The output layer uses the standard rectified linear unit (ReLU)  activation. 
An important parameter
is the  \emph{bias  $\theta_0 \in \R_+$ in the output layer}
that sets  an initial boundary level.     
Namely, the output function $g$ is given by,
$$
g(t,\xi; \theta_1,\theta_0,\theta_{in}) = \left(\theta_1 \cdot g_{in}(t,\xi; \theta_{in})+\theta_0\right)^+, 
$$
where  $g_{in}(t,\xi; \theta_{in}) \in \R^{20+\bar{d}}$ 
is the value of the neural network  before entering the output layer,
with parameters $\theta_{in}$. 
In our experiments, we have observed that the algorithm 
is not sensitive to the initial bias $\theta_0$, 
as long as it is set to a reasonable value.  For example, 
for options with a given strike $K$, 
we  choose $\theta_0=\frac{3K}{2}$,  $\theta_0=\frac{K}{2}$ for  calls 
and puts,  respectively. 
\vspace{5pt}

\noindent
{\bf{3.}}
We set the number of stochastic gradient
iterations to $M=3,000$ with a
fixed simulation batch size of $B = 512$. 
\vspace{5pt}

\noindent
{\bf{4.}} The learning rate 
process $\zeta$ 
is taken from the  Adam optimizer \cite{Kingma}.
    
\vspace{5pt}

\noindent
{\bf{5.}}
The fuzzy region width $\eps$ is chosen to be the strike $K$
times the standard deviation of the  increments of $\alpha(X)$.
While the strike $K$ adjusts to the scale of the payoff,
the second term reflects the typical variation 
of the  process $\alpha(X)$ between exercise dates. 
\vspace{5pt}

\noindent
{\bf{5.}}
After the training is completed,
the initial price is computed using the sharp boundary and 
$J= 2^{22} = 4,194,304$ many Monte Carlo simulations. 

\subsection{One-dimensional Put Option} 
\label{s.putBS}
As an initial study, we consider an at the money
 Bermudan put option  in the Black Scholes model.  Hence,
 $\varphi(t,s) = e^{-rt}(K-s)^+$, $\eta=-1$,  and $X_t=S_t$ with $\cX=\R_+$.
 
 We take the same model and parameters 
 as in  Section 4.3.1.2 of \cite{Becker2}. Namely,
 $S_t$ follows the Black-Scholes model with parameters
\begin{equation}
\label{e.p1}
s_0=K=40, \ r = 6\%,\ \sigma = 40\%,\ T=1, \ n=50,
\end{equation}
so that $\cT = \{0,\frac{1}{50},\ldots,1\}$.  
For efficient training,
we want the  the stock price to reach low values 
to cross the boundary.  Towards this goal, we use
use importance sampling with $\lambda =0.275$,
so that
$S$ becomes a super-martingale under $\Q_{\lambda}$ 
with $5\%$ negative drift. The fuzzy region width is approximately equal to 
$\epsilon = K\sigma \sqrt{T/n}$.  

It is classical that the existence of a stopping boundary Assumption
\ref{asm:star} holds with $\alpha(s)=s$ and $\Xi(s)\equiv0$.
Table \ref{tab:resultPutBS} summarizes the 
performance values
and compares it to the  values computed by Becker \emph{et.al.}~\cite{Becker1}. 
The first column in that table
is the average price of ten experiments, 
and the second column the highest price. 
In practice, one would of course retain 
the trained boundary yielding the highest price. 

\begin{table}[H]
\begin{center}
{\bf{Bermudan with parameters as in \eqref{e.p1}}}
\vspace{5pt}

\begin{tabular}{cccc}
\hline \hline
Average Price  & Highest Price & Runtime  & Price in \cite{Becker2}  \\
\hline  \hline
5.308 (0.003)  &  5.313 (0.003) &  57.9 & 5.311 \\
\hline \\[-1em]
\end{tabular}
\caption{Bermudan Put Option with $n=50$ in Black-Scholes model.
First  column is the  average  of ten experiments with its standard deviation in brackets.
The second column is the  highest price  among the realizations and its standard deviation in brackets. 
Third column is the  average runtime (in seconds) per experiment for the training phase.}
\label{tab:resultPutBS}
\end{center}
\end{table}

Figure \ref{fig:putBS} displays the initial boundary 
and the trained one.  The optimal
boundary, shown in blue, is computed
using a finite-difference scheme. 
Figure \ref{fig:putBS2} is the approximation
of the  optimal stopping boundary for the American option
with same parameters.  As the American option
does not restrict the trading dates, we  
use $n=250$.  Also,  to capture the known
singular behavior of the optimal boundary,
we use  a  
non-uniform discretization that is finer near 
maturity.

\begin{figure}[H]
\begin{center}
\includegraphics[height=2in,width=5in]{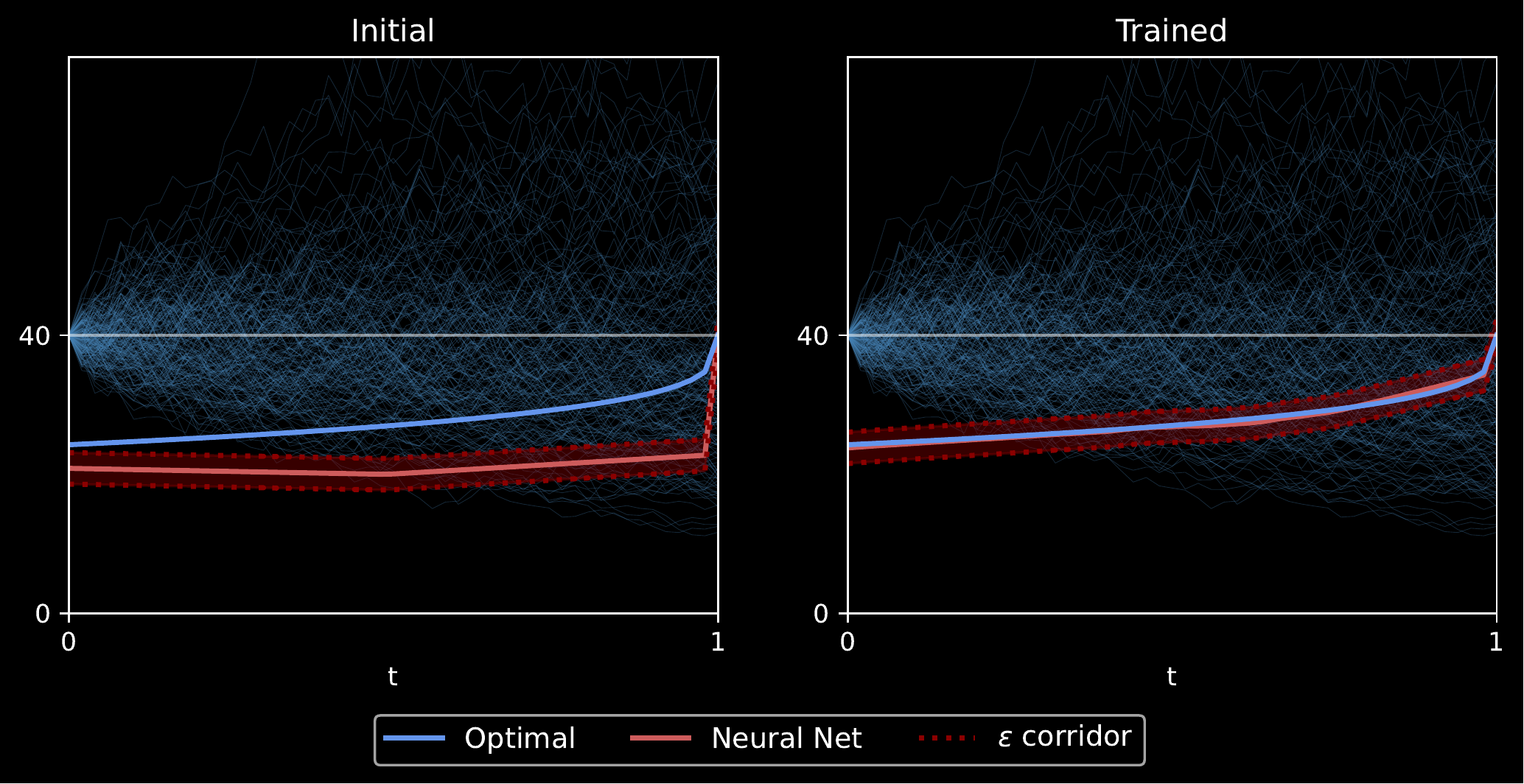}
\caption{Black-Scholes model with parameters
as in \eqref{e.p1}.
Blue trajectories are simulated paths.}%
\label{fig:putBS}
\end{center}
\end{figure}

\begin{figure}[H]
\begin{center}
\includegraphics[height=1.8in,width=5in]{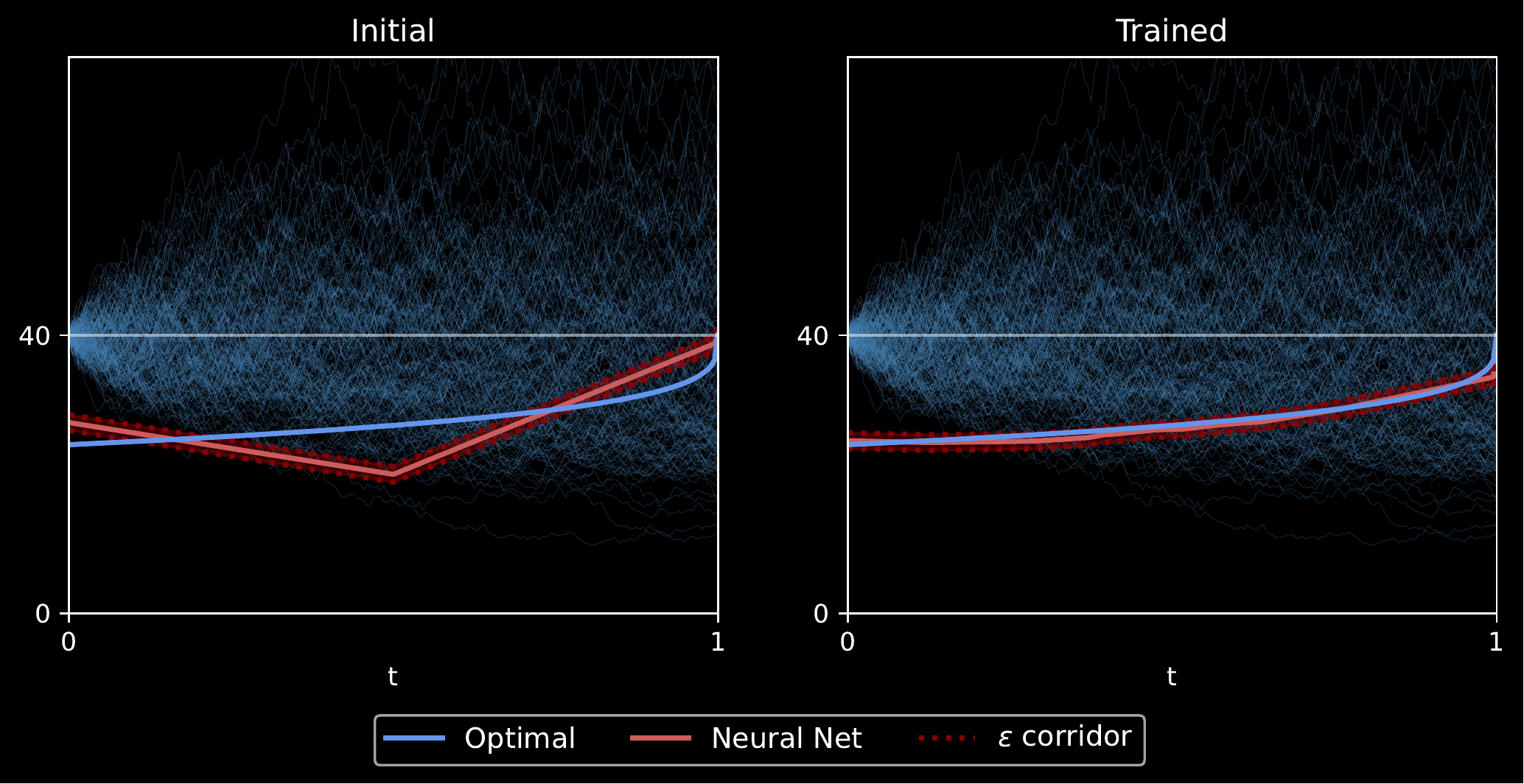}
\caption{Almost American option in
Black-Scholes model with parameters
 as in \eqref{e.p1} except $n=250$.}%
\label{fig:putBS2}
\end{center}
\end{figure}

\subsection{Put Option in the Heston Model} 
\label{s.putH}

Consider the Bermudan put option of the  Section \ref{s.putBS}
(i.e., $s_0 =K=40$, $T=1$, $n=50$) but
now in the Heston model.   
Then, under a risk-neutral measure $\Q$,
the stock and factor dynamics are given by,
\begin{align*}
\frac{dS_t}{S_t} &= r dt + \sqrt{Y_t}\  dW_t, \\
dY_t &= (\kappa(\bar{y}-Y_t) - \gamma_{\Q} Y_t)dt + 
\sigma_{Y} \sqrt{Y_t}\  d\tilde{W}_t, 
\end{align*}
where  $W, \tilde{W}$ are Brownian motions
with constant correlations  $\rho$,
and $ \gamma_{\Q}$ is the volatility risk premium
related to the risk-neutral measure $\Q$.
We use the parameters
 $$
 r = 6\%,\ \kappa = 1,\
 \sqrt{Y_0} = \sqrt{\bar{y}} = 40\%,  \ 
 \gamma_{\Q} =0,\  \sigma_{Y} =0.5,\ \rho =-0.5.
 $$ 
 In particular, 
the Feller condition $2 \kappa  \bar{y} \ge  \sigma_{Y}^2$ is satisfied.  
To allow for a comparison
with results of Section \ref{s.putBS},  
we set the initial value and the long term mean $\sqrt{\bar{y}}$ of 
the stochastic volatility $\sqrt{Y}$ 
equal to the Black-Scholes volatility $\sigma$ from
that section.   Moreover, we use importance sampling with the same 
Girsanov parameter, $\lambda =  0.275$. 
The  variance process is simulated using the Milstein scheme \cite{KP} 
 with $n' =4n=200$ time steps.   
 
 In this example, $X=(S,Y)$, $\cX=\R_+^2$,
 $\alpha(x)=s$, $\beta \equiv 0$, and $\Xi(x)=(1,y)$.  We simplify
and set  $\Xi$  to $y$.
Figure \ref{fig:heston1} shows that the stopping boundary is a function
of time $t$ and the spot variance $Y_t =y$.  
As $\varphi$ is convex and bounded,  
 the value function $v(t,s,y)$
is  non-decreasing in $y$ for all $t\in [0,T]$, cf.~\cite{LambertonHeston}.
Consequently, the map $y \mapsto f(t,y)$ is non-increasing
and the trained neural network $g(\cdot,\cdot;\theta^M)$
has the same property as seen
 in  \cref{fig:heston1}.
Notice also that $y \mapsto f(t,y)$ becomes less steep as $t$ 
approaches maturity. This is line with the fact that the 
vega of the option  decreases over time prior to the exercise date.  
Since $s \mapsto v(t,s,y)$ is non-increasing for fixed $(t,y)$,  
we therefore expect that the  rectangle $[0,s]\times [0,y]$ 
is contained in $\calS_t$ whenever $(s,y) \in \calS_t$. 
This is also confirmed in  \cref{fig:heston1}. 
\vspace{10pt}
 
\begin{figure}[h]
\centering
\includegraphics[height=1.8in,width=5in]{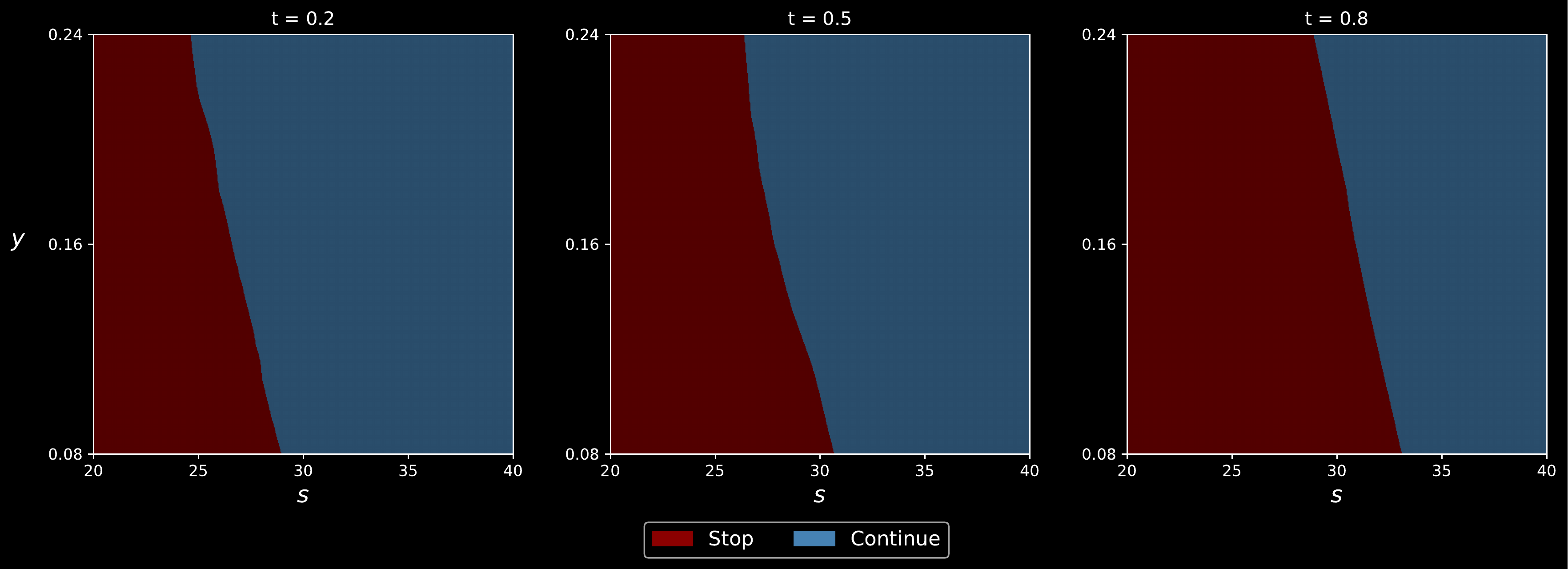}
\caption{Stopping and continuation regions for a Bermudan put in the Heston model.}
\label{fig:heston1}
\end{figure}
 Table \ref{tab:resultPutHeston} reports the prices 
 of the claim.  Note that the longer runtime compared to Table \ref{tab:resultPutBS} 
 is due to the thin partition needed to simulate the spot variance process $Y$.
\begin{table}[H]
\centering 
{\bf{Bermudan Put Option ($n=50$), Heston model}}
\vspace{3pt}

\begin{tabular}{c c c c }
\hline  \hline
  Average Price& Highest Price & Runtime  \\
\hline  \hline
 5.033 (0.003) & 5.037 (0.003) &  201.0 \\
\hline\\[-1em]
\end{tabular}
\caption{First column is the
average  of ten experiments with its standard deviation in brackets. 
Second column is the highest price  
among the realizations and its standard deviation in brackets.  
Third column is the  average run-time (in seconds) per experiment for the training phase.}
\label{tab:resultPutHeston}
\end{table}

Finally, \cref{fig:heston2} displays a trajectory of $(X,Y)$ together with the 
trained stopping boundary
 $g(\cdot;\theta^M)$. As can be seen, the threshold is typically below the 
 Black-Scholes boundary (red curve) when $Y_t$ is above its mean $\bar{y} = 0.16$ and vice versa. 

\begin{figure}[h]
\centering
\includegraphics[height=1.8in,width=3.4in]{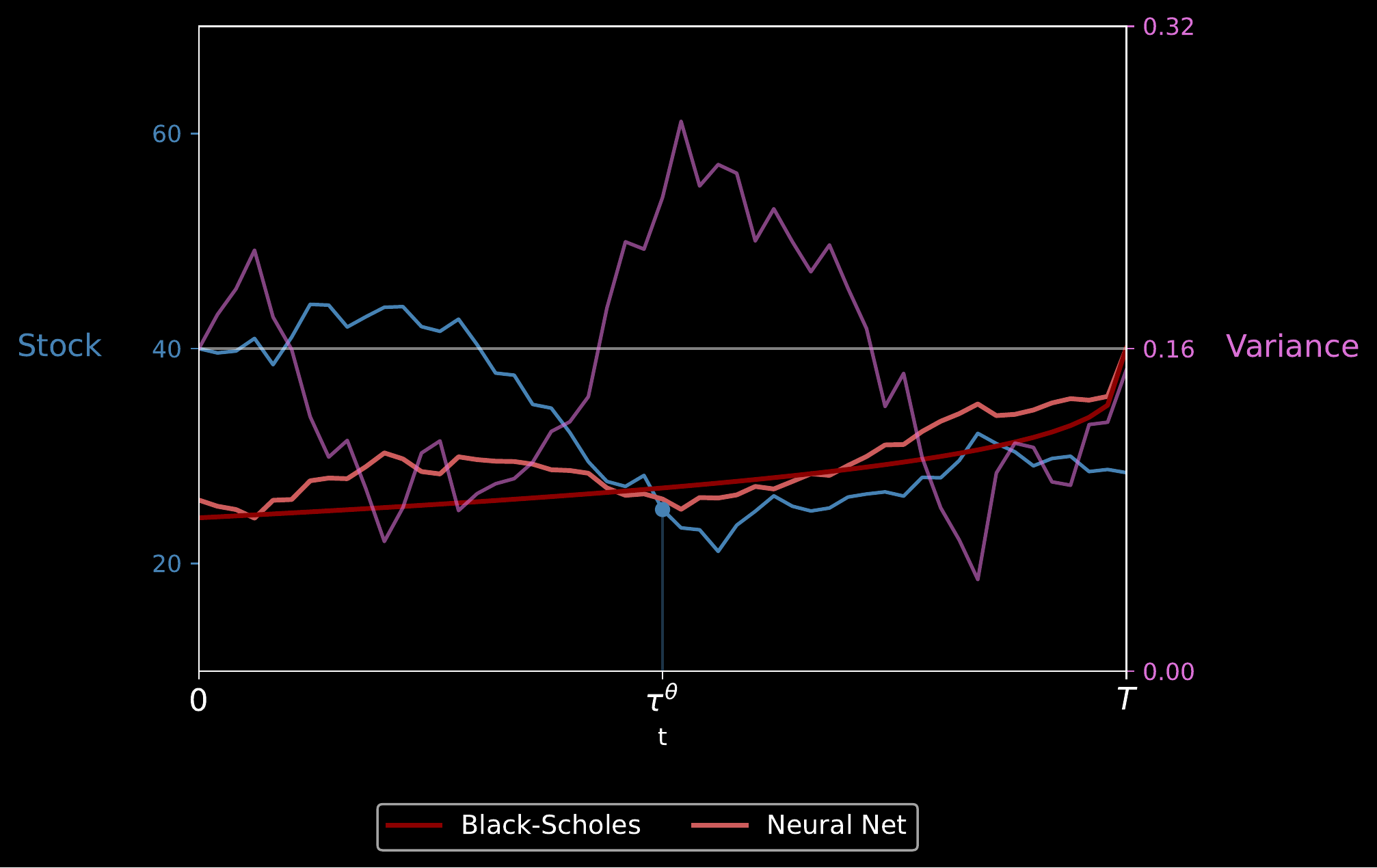}
 \caption{Stock price, variance and threshold process for a Bermudan put in the Heston model.}
\label{fig:heston2}
\end{figure}

\subsection{Bermudan Straddle}
\label{s.straddle}
Consider a Bermudan straddle  in the Black-Scholes model, that is $\varphi(t,s) = e^{-rt}|s-K|^+$, $X_t =S_t$, and $\calX = \R_+$. As already seen in Example~\ref{ex.straddle.coordinates},  the payoff is \textit{not} of the form \eqref{eq:payoff}. The purpose of this section  is therefore to demonstrate the flexibility of the method. 
 As argued in Example~\ref{ex.straddle.coordinates}, there exist \textit{two} boundaries $f^U,f^L: \calT^{\circ}\to [0,\infty]$ such that    
 $\calS = \calS^L \cup \calS^U,$ where  
 $\calS^{\text{U}}$ (respectively $\calS^{\text{L}}$)   is the epigraph of $f^U$ (resp. hypograph of $f^L$).

We consider the same Black-Scholes  parameters as in \eqref{e.p1}, except that the stock has  a dividend rate of $\delta = 0.06$. This is to prevent that the option is never exercised prior to maturity when the underlying is above the strike, in which case $f^U \equiv \infty$. The 

In the implementation, we reconfigure the neural network with a two-dimensional output and  train the lower and upper boundaries simultaneously. Because of the symmetry of the problem, it is reasonable to make the the underlying process a martingale so as to visit both high and low values. Since $r = \delta = 0.06$,  the drift is already  zero under $\Q$ so importance sampling is not used in this case. 
The results are given in Table \ref{tab:resultStraddleBS} and the two boundaries illustrated in \cref{fig:straddleBS}. As can be seen, the neural network can effectively handle multiple boundaries.  

\begin{table}[H]
\begin{center}
{\bf{Bermudan Straddle $(n=50)$, Black-Scholes model}}
\vspace{5pt}

\begin{tabular}{cccc}
\hline \hline
Average Price  & Highest Price & Runtime  & LSMC Price \\
\hline  \hline
12.080 (0.004)  &  12.087 (0.004) &  24.6 & 12.018 (0.004) \\
\hline \\[-1em]
\end{tabular}
\caption{Bermudan at-the-money Straddle with $n=50$ in Black-Scholes model.
First  column is the  average  of ten experiments with its standard deviation in brackets.
The second column is the  highest price  among the realizations and its standard deviation in brackets. 
Third column is the  average runtime (in seconds) per experiment for the training phase and fourth column is the price obtained with the Least square Monte Carlo algorithm.}
\label{tab:resultStraddleBS}
\end{center}
\end{table}

\begin{figure}[H]
\begin{center}
\includegraphics[height=2in,width=5in]{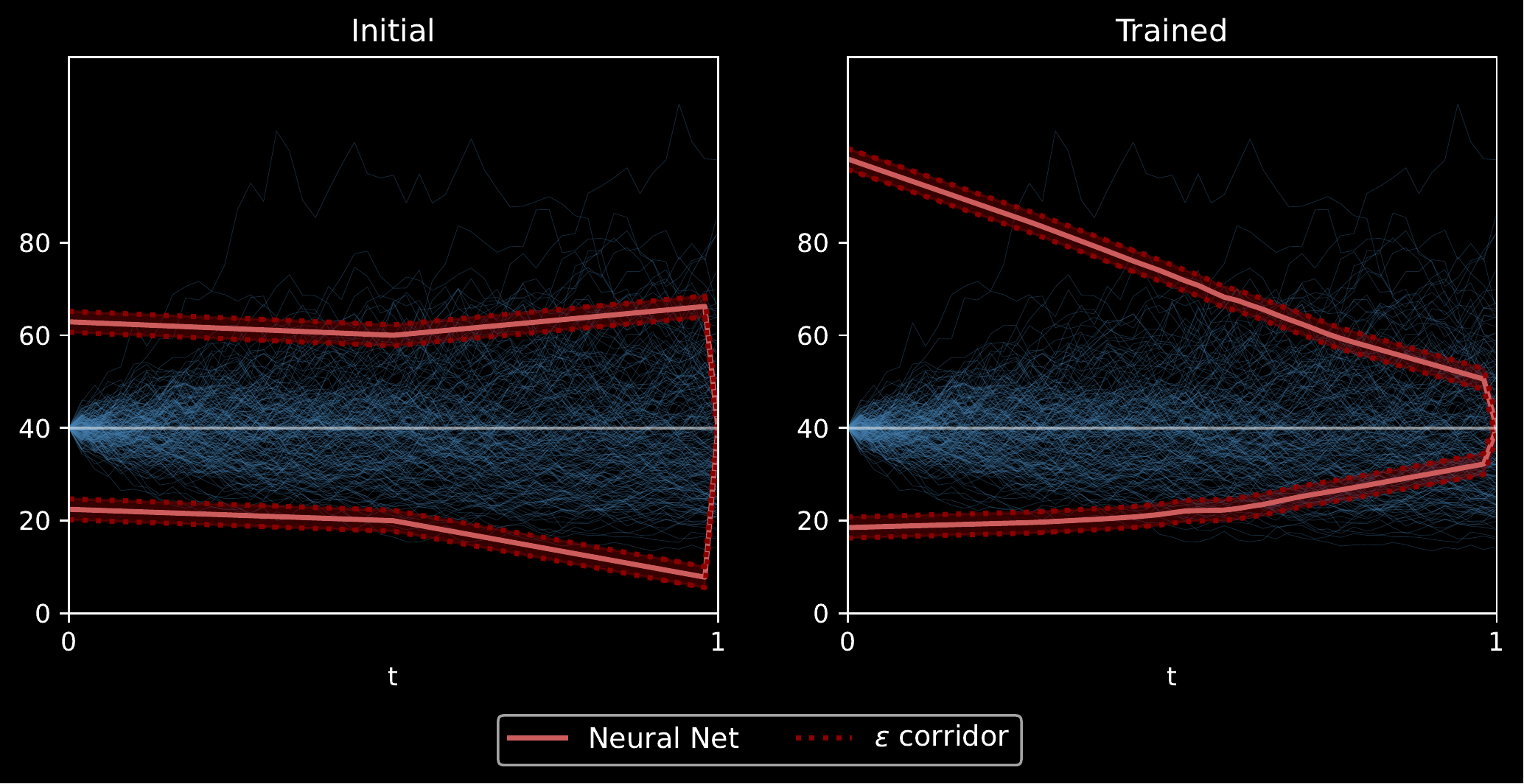}
\caption{Bermudan straddle ($n=50$) in the Black-Scholes model.}%
\label{fig:straddleBS}
\end{center}
\end{figure}


\subsection{Bermudan Max-Call}
\label{s.maxcall}

We now apply the deep empirical risk minimization
to the classical example  of the max-call option
on dividend paying stocks
as studied in the seminal paper by
Broadie \& Detemple \cite{BD},
and also in the book by Detemple  \cite{Dbook}.

We assume that
the stock price process
$S_t =(S^{(1)}_t,\ldots,S^{(d)}_t)\in \R_+^d$ 
is the state process and solves
$$
dS^{(i)}_t = S^{(i)}_t \Big ( (r-\delta_i)dt + \sigma_i dW^{(i)}_t\Big),
$$
where $\delta_i$ is the dividend rate. 
In this example, the reward function  is given by,
$$
\varphi(t,s) = e^{-rt} (\alpha(s) -K)^+, 
\quad
\text{where}
\quad\alpha(s) = \max\{ s_1,\ldots,s_d\}, \qquad 
s \in \R_+^d.
$$
Hence, $X=S$, $\cX=\R_+^d$, $\eta=1$, and $\beta \equiv 0$.
In view of Theorem \ref{t.fb}, the existence of a 
stopping boundary Assumption \ref{asm:star}
holds with $\alpha$
and $\Xi(s)=s/\alpha(s)$.
In two dimensions, the stopping 
regions can be visualized effectively
as seen in the Figures \ref{t.2}, \ref{t.3} 
reported from \cite{RST}.
These are stopping regions
in two space dimensions
obtained with 
initial data $s_0=90$ and $s_0=100$
with parameters  given in \eqref{e.p2}.
Clearly the stopping boundary is independent of
the initial condition
and the below numerical results
verify it.  Also they are very
similar to those obtained in \cite{BD}.

\begin{figure}[H]
\centering
\includegraphics[height=1.75in,width=5in]{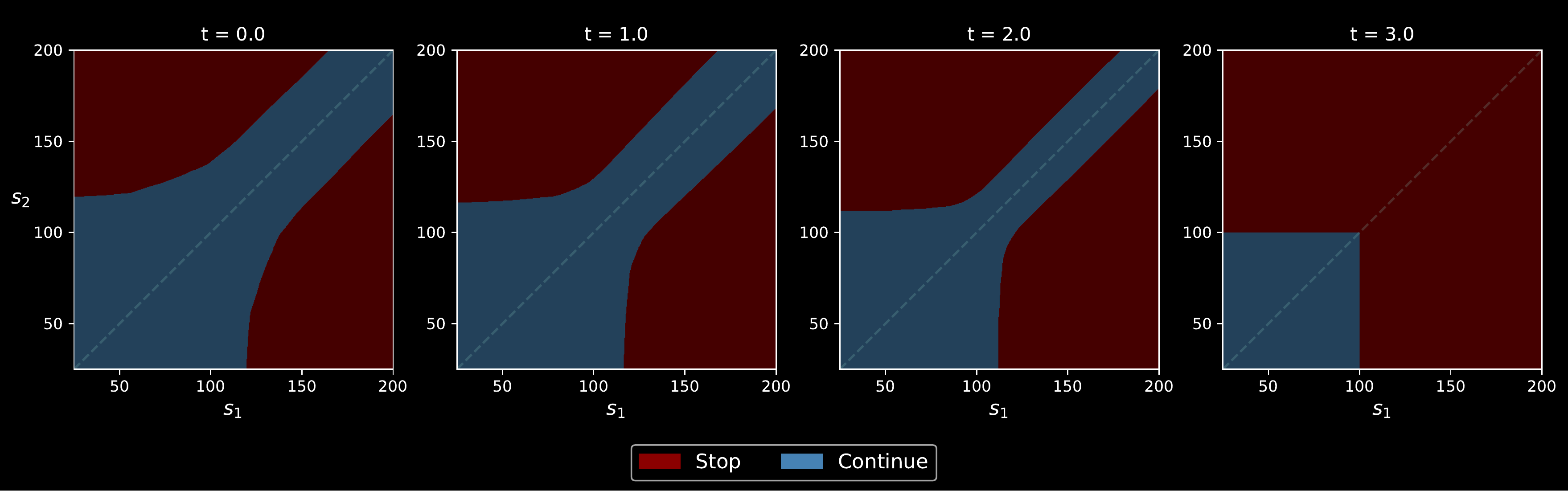} 
\caption{Stopping Boundaries with $s_0=90$, $K=100$, 
$\sigma_i =0.2$, $r =0.05$, $\delta_i =0.1$, reported from \cite{RST}.}
\label{t.2}
\end{figure}

\vspace{-3mm}

\begin{figure}[H]
\begin{center}
\includegraphics[height=1.75in,width=5in]{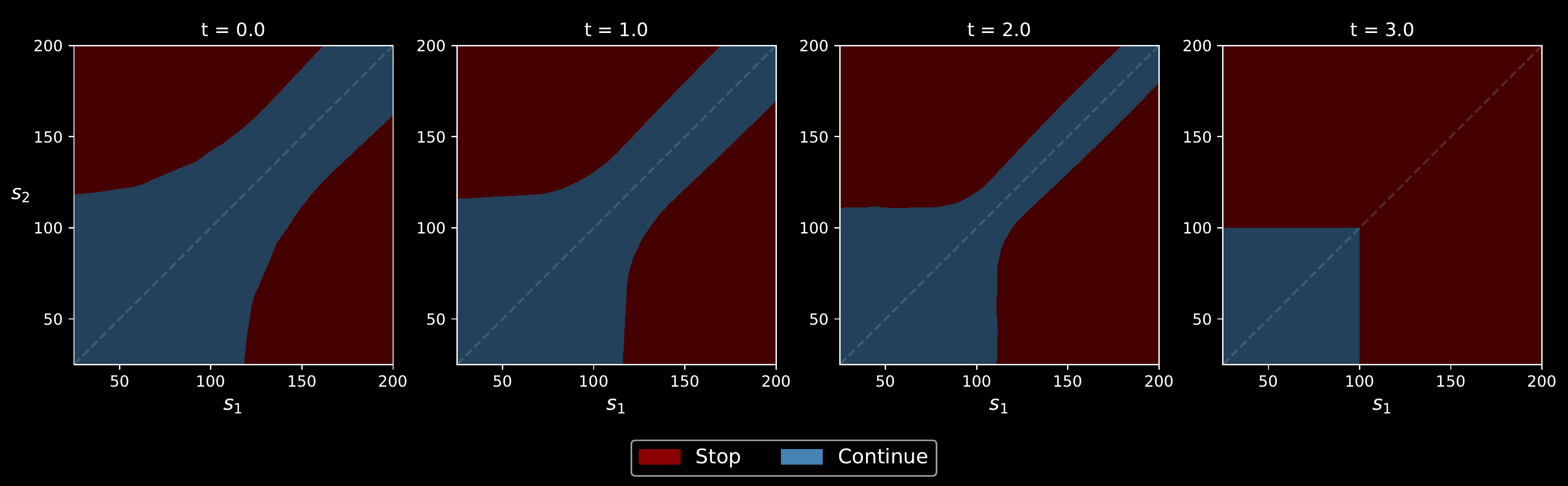}
\caption{Stopping Boundaries with $s_0=100$, $K=100$,
$\sigma_i =0.2$, $r =0.05$, $\delta_i =0.1$, reported from \cite{RST}.}
\label{t.3}
\end{center}
\end{figure}

\subsubsection{Symmetric Max-call}
\label{ss.num-mc-sy}
As in \cite{Becker1,Becker2}, we use the following parameter set:
\begin{equation}
\label{e.p2}
K=100,\ s_0= 100,\ \sigma_i =0.2,\ r =0.05,\ \delta_i =0.1.
\end{equation}
We take the maturity to be $3$ years and $n=9$.  Thus, each
time interval corresponds to four months. Numerical experiments
for these parameters are reported
in the accompanying paper \cite{RST}.
They compare well to the results obtained in \cite{Becker1,Becker2}.

We also compute the initial price of this max option on 
$d \in\{2,5,10,20,50\}$ assets.
Since we observe that the drift of the maximum stock
price process is of order $\ln(d)$, we 
use important sampling with $\lambda_i=(r-\delta_i-\mu(d))/\sigma$
and $\mu(d)= - 0.01\ln(d)$ to ensure that
the maximum stock value does not cross the stopping boundary too soon.
Utilizing the fact that the stocks
are exchangeable in this example, we order the stock prices $s_{(1)} \ge s_{(2)} \ge \ldots $ and
input the second to sixth largest ratios $(\frac{s_{(2)}}{\alpha(s)},\ldots, \frac{s_{(6)}}{\alpha(s)})$ to the neural network (note that the first ratio $\frac{s_{(1)}}{\alpha(s)} \equiv 1$ and is therefore omitted).   

 Table \ref{tab:symMaxOpt} summarizes the results. 
 Despite a low cutoff, the  obtained prices are close or above the benchmark.
 Moreover, the runtime remains moderate as $d$ increases. 
 The case $d=5$ is a classical example first 
 introduced by \citet{BroadieGlasserman} and well-studied later in the literature. 
Table \ref{tab:symMaxOpt} also contains the tight confidence intervals, obtained in \cite{AB}, \cite{BC}, and \cite{Becker1}
using a primal-dual method. 

\begin{table}[h]
\centering 
{\bf{Max-call option with parameters in \eqref{e.p2} and $s_0=100$}}
\vspace{5pt}

\centering
\begin{tabular}{ c c c c c c}
\hline  \hline
$d$ & Average Price& Highest Price & Runtime & Price in \cite{Becker2} & 
Confidence Intervals \\ 
\hline \hline
2   & 13.883 (0.009)  &  13.898 (0.008) & 29.1  & 13.901 & \ [13.892, 13.934] \; (\cite{AB}) \\
5   & 26.130 (0.010)  & 26.151 (0.009) & 31.8  & 26.147 &  \ [26.115, 26.164] \; (\cite{BC})\\
10   & 38.336 (0.015) &   38.355 (0.011)& 33.1 & 38.272 & \ [38.300, 38.367]  \; (\cite{Becker1})\\
20   & 51.728 (0.018) & 51.753 (0.011)  & 36.0 & 51.572  & \ [51.549, 51.803]
 \; (\cite{Becker1})\\
50  & 69.860 (0.012) & 69.881 (0.011) & 43.5 & 69.572   &  \ [69.560, 69.945]
 \; (\cite{Becker1})\\
\hline\\[-1em]
\end{tabular}
\caption{Max Option on $d \in\{5,10,20,50\}$ symmetric assets.
The second column is the  average  of ten experiments with its standard deviation in brackets. The third column is the highest price  among the realizations and its standard deviation in brackets.  
The fourth column is the  average runtime (in seconds) per experiment for the training phase.}
\label{tab:symMaxOpt}
\end{table}

\subsubsection{Asymmetric Max-call}
\label{ss.num-mc-nons}

We next investigate the Bermudan max option
with  asymmetric assets. We consider two 
stocks with same volatility $\sigma$ but different dividend rates. 
We use the same parameters for $K,r,\sigma$
as in  the previous subsection, and 
$s_0=100$, $\delta = (5 \%,15 \%)$. 
 
 We choose the parameter
 $\lambda$ in the Girsanov theorem,
 so as to make the assets symmetric under an equivalent measure. 
 More precisely, we choose 
 $\lambda = [(r+\mu(2))(1,1) - \delta]/{\sigma} \in \R^2$ 
 so that both assets have the same  drift $\mu(2) \in \R$ under $\Q_{\lambda}$,
 and $\mu(2)=-0.01\ln(2)$ is as in the previous subsection.  
 The asymmetry is then captured by the Radon-Nikodym derivative 
 appearing in the reward function.

The average and highest price after $10$ realizations 
are given in Table \ref{tab:asymMaxOpt}. The last column  is a 
benchmark price obtained from the Least Square 
Monte Carlo (LSMC) algorithm \cite{LS}  with $2^{22}$  simulations. 
As can be seen, our method in average computes a price comparable  
to the benchmark, but may construct better stopping strategies 
reported in the second column.
Notice also that having different dividends gives a premium
 over the symmetric case as shown in the first row of Table \ref{tab:symMaxOpt}.
 
\begin{table}[h!]
\centering 
{\bf{Max-call option with asymmetric assets}}
\vspace{5pt}

\centering
\begin{tabular}{ c c c c }
\hline  \hline
 Average Price& Highest Price & Runtime & LSMC Price\\
\hline \hline \\[-1em]
15.551 (0.014)  & 15.575 (0.011) & 30.9 & 15.558 (0.009) \\
\hline\\[-1em]
\end{tabular}
\caption{Max Option on $d=2$ asymmetric assets}
\label{tab:asymMaxOpt}
 \end{table}
  
\cref{fig:asymCall} displays the stopping and continuation
region over time. 
In particular, we clearly see the assymmetry of the problem and the stopping region reflected in the figure.
For $i=1,2$, let the 
the connected components of $\calS_t$
be given by,
$\calS^{(i)}_t := \left\{ s \in \cS_t 
:  \alpha(s) = s_i \right\}$.
In \cref{fig:asymCall}, we observe that
the light red region $\cI(\calS^{(1)}_t)$
with $\cI(s_1,s_2) := (s_2,s_1)$,
is contained in $\calS^{(2)}_t$. 
That is, if it is optimal to stop
at $S_t = (s_1,s_2)$, 
then the same is true at   $\cI(S_t)=(s_2,s_1)$. 
The structure of the assymmetry is therefore consistent with our expectations.

\begin{figure}[h]
\centering
\includegraphics[height=1.75in,width=5in]{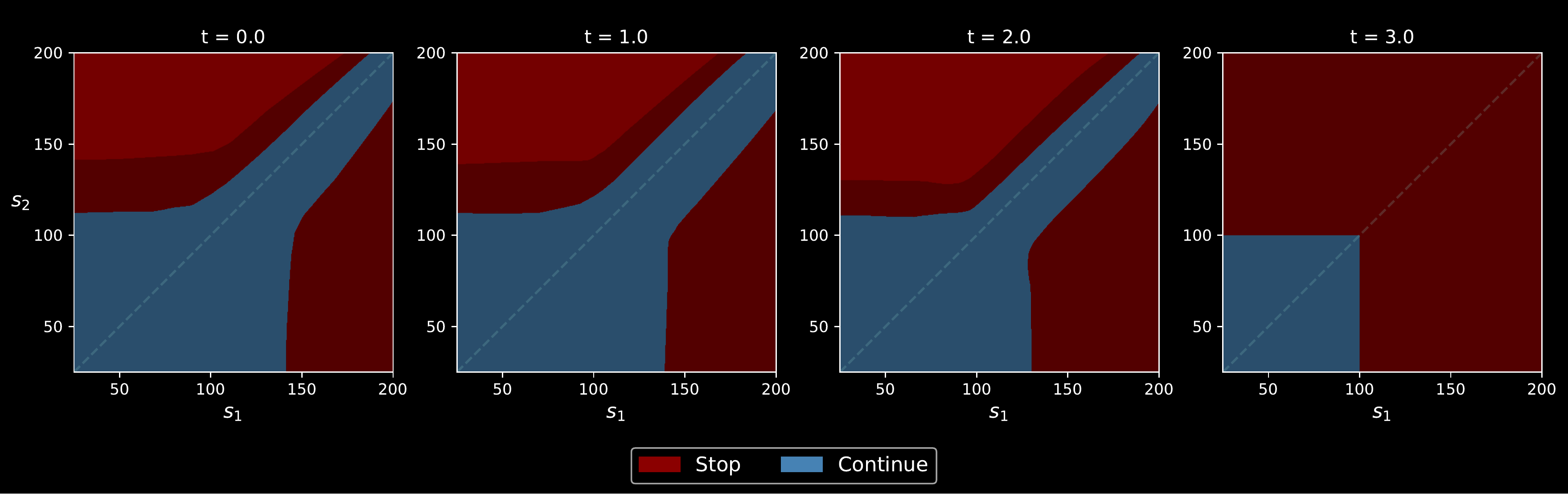} 
\caption{Stopping and continuation regions  for a max option with
$d=2$ assets and $T=3$.
The light red region is the reflection of the lower stopping region $\{ s \in \cS_t\ 
: \ s_1 \ge s_2\}$ through the diagonal.}
\label{fig:asymCall}   
\end{figure}

\begin{remark}
\label{r.others}

For 
max-call options on two assets,  \cite{Garcia}
also use a two-dimensional parametrization
of the stopping region given by
$$
\cS_t \approx \left\{ (s_1,s_2)\in \R_+^2\ :\
\alpha(s)=\max\{s_1,s_2\} \ge \theta_t^{1}, \ \text{and}\ \ 
|s_1-s_2| \ge \theta_t^{2}\ \right\}.
$$
While the above hypothesis class of regions
provides satisfactory numerical results, 
they are restricted to two dimensions,
and only give a  simple polygonal approximation 
of the actual stopping regions
that are more complex as can be seen in the 
Figures \ref{t.2}, \ref{t.3}, \ref{fig:asymCall}.
\end{remark}

\subsubsection{Up-and-Out Max-call} 
\label{s.UOBarrier}
We consider a Bermudan \textit{up-and-out} max-call option on $d \ge 2$ symmetric assets. That is, the payoff is given by 
\begin{equation*}
    \varphi(X_t) = (\alpha(S_t) - K)^+ \mathds{1}_{\{Z_t \le b\}}, \quad X_t = (S_t,Z_t) \in \R^{d+1}_{+}, \quad (\alpha(s) = \max\{ s_1,\ldots,s_d\})\\ 
\end{equation*} 
where $b$ is the \textit{barrier} and $Z_t = \overline{S}_t:=\max_{u\in[0,t]}  \alpha(S_u)$ is the running maximum across all assets.  
We take the  parameters from \cite{Moallemi}, namely
\begin{equation*}
\label{e.UOBarrier}
K=100,\ b = 170, \ s_0= 100,\ \sigma_i =0.2,\ r =0.05,\ \delta_i =0, \ T = 3, \ N = 54. 
\end{equation*}
In this example, 
we set $\Xi(x) = (s/\alpha(s),z)$ instead of $\Xi(x) = x/\alpha(s), \ x = (s,z)$. In other words, the raw running maximum is given to the boundary instead of the ratio $\bar{S}_t/\alpha(S_t)$ to directly compare $\bar{S}_t$ with the barrier level. We use the same cutoff (5 assets) and Girsanov parameter as in  \cref{ss.num-mc-sy}. 
The results are reported in \cref{tab:UOMax}. As can be seen, the average and highest  prices (first two columns) lies within the price intervals from \cite{Moallemi} (last column). Notice that the price of the up-and-out max-call option increases slowly with the number of assets ($d$) because the contract becomes more likely to be knocked out. Indeed, the drift of the running maximum process $\bar{S}$ increases with $d$, increasing the probability of hitting the barrier $b$. 
\begin{table}[h]
\centering 
{\bf{Up-and-Out Max-call option with parameters in \eqref{e.p2} and $s_0=100$}}
\vspace{5pt}

\centering
\begin{tabular}{ c c c c c }
\hline  \hline
$d$ & Average Price& Highest Price & Runtime & Price Interval from \cite{Moallemi}  
\\ 
\hline \hline
4   & 42.487 (0.010)  & 43.093 (0.009)& 49.2  &  \ [41.541, 43.853] \\
8   &  50.905 (0.007) & 51.379 (0.006) & 57.2  &   \ [50.252, 52.053] \\
16   & 53.650 (0.007) &  54.468 (0.006) & 62.7 & \ [53.638, 55.094]  \\
\hline\\[-1em]
\end{tabular}
\caption{Up-and-out max-call option on $d \in\{4,8,16\}$ symmetric assets.
The second column is the  average  of ten experiments with its standard deviation in brackets. The third column is the highest price  among the realizations and its standard deviation in brackets.  
The fourth column is the  average runtime (in seconds) per experiment for the training phase.}
\label{tab:UOMax}
\end{table}

\subsection{Look-back Options} 
\label{s.look-back}

\emph{Look-back options}  provides exposure to the 
minimum or maximum values attained 
during the tenure of the option.
There are several types of look-back options 
that are commonly traded as summarized
in the Table \ref{tab:lookback}.   In these examples,
when the stock price 
$S_t \in \R_+$ is modeled as a Markov process, 
then we would have $X=(S, Z)$, where $Z_t$ is either
the \emph{running maximum} $\overline{S}_t$ or 
\emph{running minimum} $\underline{S}_t$ given by,
$$
\overline{S}_t:=\max_{u\in[0,t]} S_u ,\qquad
\underline{S}_t:=\min_{u\in[0,t]} S_u .
$$
\begin{table}[H]
\centering
\begin{tabular}{lcrcc}
\hline \hline
\noalign{\vskip 0.5em}
Option & $e^{rt}\varphi(t,X_t)$ & $\eta$ & $\alpha(s,z)$ & $\beta(s,z)$ \\[0.5em] \hline \hline
\noalign{\vskip 0.5em}
Fixed Strike Call & $(\overline{S}_t - K)^{+} $ & $1$ & $z$ & $0$ \\[0.5em] \hline
\noalign{\vskip 0.5em}
Fixed Strike Put & $(K - \underline{S}_t)^{+} $ & $-1$ & $z$ & $0$ \\[0.5em] \hline
\noalign{\vskip 0.5em}
Floating Strike Call & $(S_t - \gamma \underline{S}_t)^+$  & $1$ & $s$ & $\gamma  z$ \\[0.5em] \hline
\noalign{\vskip 0.5em}
Floating Strike Put & $(\gamma \overline{S}_t  - S_t)^+$ & $-1$ & $s$ &  $\gamma z$ \\[0.3em] \hline 
\end{tabular}
\caption{Different Lookback Options.}
\label{tab:lookback}
\end{table}

Clearly, $X=(S,Z)$ is a  Markov process with
$\cX = \{x=(s,z) \in \R_+^2\ | \ 0< s \le z\}$, when $Z_t=\overline{S}_t$,
or $\cX = \{x=(s,z) \in \R_+^2\ | \ 0< z \le s\}$ when $Z_t=\underline{S}_t$.
The scaling factor $\gamma$ appearing in the payoff of floating strike 
options is introduced to reduce the price of these otherwise expensive contracts. 
We therefore choose $\gamma \in [1,\infty)$ and $\gamma \in (0,1]$ for call and put options, respectively.  
When $\gamma=1$, the floating strike look-back put (call) option delivers precisely the
drawdown (drawup) of the stock. 
We refer the reader to \citet{DaiKwok} 
for an introduction and discussion of American look-back claims.  

\subsubsection{Fixed strike call}
\label{ss.fsc}

Although our setting
covers all look-backs, as an example we only consider
the fixed strike call.  Then, 
$Z_t = \overline{S}_t$,
$\cX = \{x=(s,z) \in \R_+^2 \ | \  0 < s \le z\}$),
$\eta=1$, $\alpha(s,z) = z$.    
By Theorem~\ref{t.fb}, a stopping boundary exist with 
$\Xi(x) = (s/z,1)$.  
Notice that 
$1-(s/z) = (z-s)/z$ is 
 the relative drawdown of the stock. As in \cite{DaiKwok}, we use
\be\label{e.p3} 
T=1/2, \quad  r=2\%, \quad \delta =4\%, \quad  \sigma = 30 \%.
\ee 
To approximate the American option,
we consider a regular partition  with $n=200$ many exercise dates. 
Moreover, we employ an even thinner  partition with $n'=800$ points 
for our simulations.
This allows us to better approximate  the running maximum. 
Importance sampling is not needed in this example.
We visualize the the stopping region in the $(s,y)$ plane, similar to Fig.~3 in \cite{DaiKwok},
where the authors use a finite difference scheme to compute the initial price.

\begin{table}[H]
  \centering

  \begin{tabular}{ccccc}
 \hline \hline
    Average Price  & Highest Price & Runtime   & European Price & Upper Bound\\
  \hline  \hline \\[-1em]
  16.827   (0.008)  &  16.844 (0.007) & 216.7  &  16.808 & 16.979 \\
  \hline 
\end{tabular}
  \caption{American look-back call with parameters in \eqref{e.p3} and $K=100$. }
\label{tab:resultLkbk}
  \end{table}
  
Table \ref{tab:resultLkbk} summarizes the result for $s_0=K=100$. 
 First  column in that table is the average of ten experiments 
 with the standard deviation in brackets. The second column is the highest price  
 among the realizations and its standard deviation in brackets.  
 The third column is the  average runtime (in seconds) per experiment for the training phase. 

 The last column is the upper bound given by the forward value of the European price.
The European price, denoted by $v_e$ provides a 
lower bound on the American price. 
The  last  column is the  forward value of European price,
 $e^{rT}v_e$, giving an upper bound for the American price.
 Since here $e^{rT} \approx 1.01$, the American option only offers little premium 
 over its European counterparts. This is nevertheless  fortunate from a 
 performance perspective as it gives a tight interval in which the American price 
 must lie (see, e.g., \cite{Conze}).  We indeed obtain prices that are within the bounds. 

 \begin{figure}[H]
\centering
\includegraphics[height=1.65in,width=5in]{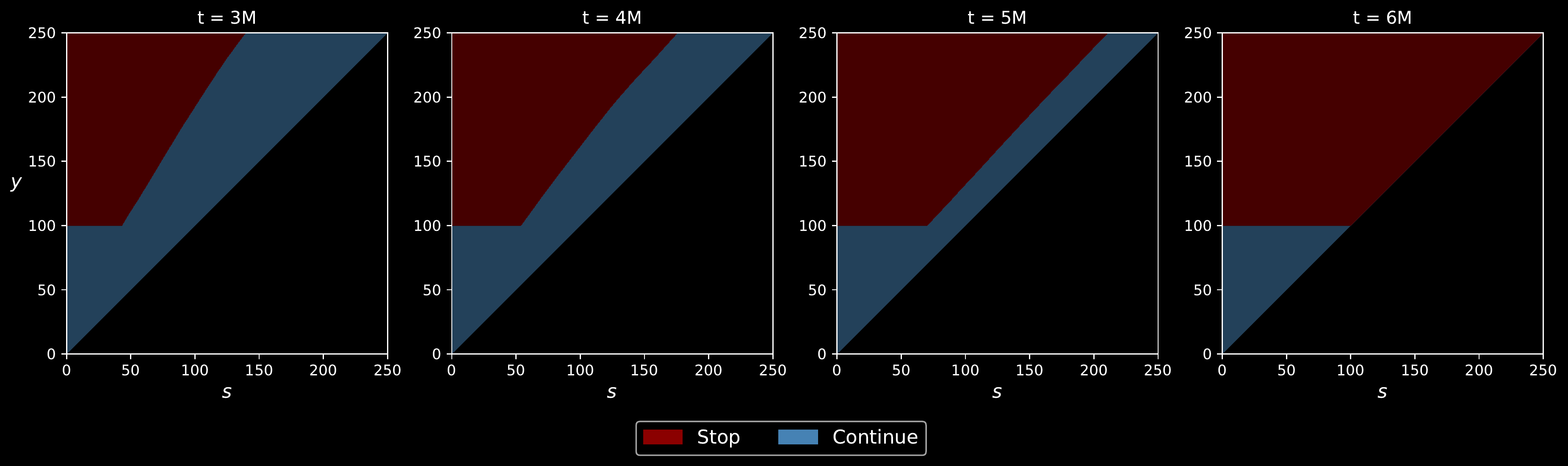}
\caption{Stopping boundaries for a look-back call  with parameters in \eqref{e.p3} and $K=100$.}
\label{fig:lkbkCall}
\end{figure}

We display the stopping and continuation region in \cref{fig:lkbkCall}. 
Obtaining an overall accurate boundary turns out to be challenging,
as visiting all the pairs $(s,y) \in \calX$ is difficult, especially when $s \ll y$. 
We effectively resolve this issue
by randomizing $s_0\in \calX$ and setting the initial running maximum to 
$y_0 = K \vee s_0$. As in Fig. $3$ of \cite{DaiKwok}, we observe a flat boundary when $s\ll K$. 
That is, the neural network correctly set the boundary  equal to the strike when 
the ratio $s/y$ is small, or equivalently, when the relative drawdown is large.

\section{Conclusion}

We have developed an algorithm for the
computation of the ``graph-like'' stopping boundaries
separating the continuation and stopping regions
in optimal stopping.  While the method at the high-level is 
empirical risk minimization, a relaxation 
based on fuzzy boundaries motivated from phase-field 
models for liquid-to-solid phase transitions is used.
Through numerical experiments, the method is shown 
to be effective in high-dimensions.  The method has the potential
to incorporate market details like price impact, transaction costs,
and market restrictions.  It is also possible to apply the technique
to other models from financial economics
and other obstacle type problems as long as a 
control representation is available.

\appendix 

\section{Proof of Theorem \ref{t.fb}}
\label{a.existence}

We prove the theorem when the parameter $\eta$ in \eqref{eq:payoff}
is equal to $-1$. The other case of $\eta=1$ is proved \emph{mutatis mutandis}.
We fix $t \in \cT^\circ$.

Firstly, it is clear that for $\Xi$ as in \eqref{e.xi},
the map $x=(s,y,z) \mapsto (\alpha(s,z), \Xi(x))$ is a
homeomorphism and $\alpha$ is onto. 

Recall that $\vartheta$ is the set of all
$\cT$-valued stopping times, $\vartheta_t=\vartheta\cap[t,T]$,
the pay-off $\varphi(t,x)=\varphi(t,s,z)$ is given by \eqref{eq:payoff},
and the European price, $v_e$, is strictly positive.
As the European price is equal to price of
the stopping time $\tau\equiv T$,
$v(t,x) \ge v_e >0$.  Therefore,   the stopping region is 
characterized by
$$
x \in \cS_t
\quad
\Leftrightarrow
\quad
v(t,x) =\varphi(t,s,z)>0.
$$

We start by showing that 
$\cS_t$ is star-shaped around the origin.
\begin{lemma}
\label{l.geom}
If $x=(s,y,z)  \in \cS_t$, then $(\gamma s, y, \gamma z) \in \cS_t$
for every $\gamma \in (0,1]$.
\end{lemma}

\begin{proof}
Fix $x=(s,y,z)  \in \cS_t$ and $\gamma \in (0,1]$.
Then, by the definition of the stopping region $\cS_t$,
$v(t,x) =\varphi(t,s,z)>0$. In particular, as $\eta=-1$,
$\alpha(s,z)-\beta(s,z) < K$ and therefore, 
$$
\alpha(\gamma s,\gamma z) -\beta(\gamma s,\gamma z)
=\gamma(\alpha(s,z)-\beta(s,z) ) < \gamma K \le K.
$$
Hence, $\varphi(t,\gamma s,\gamma z) = - (\alpha(\gamma s,\gamma z) -\beta(\gamma s,\gamma z)-K)>0$.
Moreover,
$$
\E\left[ \varphi(\tau,S_\tau,Z_\tau) \mid  X_t =(\gamma s, y, \gamma z)\right]
= \E\left[ \varphi(\tau,\gamma S_\tau,\gamma Z_\tau) \mid  X_t =(s,y,z)\right].
$$
for all  $t \in \cT^\circ$, $x=(s,y,z)$, $\tau \in \vartheta_t$, and $\gamma>0$.
By \eqref{eq:payoff} and the homogeneity of $\alpha$ and $\beta$,
\begin{align*}
\varphi(\tau,\gamma S_\tau, \gamma Z_\tau) &= 
e^{-r\tau} \left(-(\alpha(\gamma S_\tau, \gamma Z_\tau)
- \beta(\gamma S_\tau, \gamma Z_\tau) -K)\right)^+\\
&= e^{-r\tau} \left(-(\gamma \alpha(S_\tau,Z_\tau)
- \gamma \beta(S_\tau,Z_\tau)-K)\right)^+\\
&\le e^{-r\tau} \gamma \left(-( \alpha(S_\tau,Z_\tau)
-  \beta(S_\tau,Z_\tau)-K)\right)^+  + e^{-r\tau } (1-\gamma) K\\
& = \gamma \varphi(\tau,S_\tau,Z_\tau) + e^{-r\tau } (1-\gamma) K.
\end{align*}
We combine the above inequalities to arrive at
the following:
\begin{align*}
\E\left[ \varphi(\tau,S_\tau,Z_\tau) \mid  X_t =(\gamma s, y, \gamma z)\right]
& \le \gamma \E\left[ \varphi(\tau, S_\tau,Z_\tau) \mid  X_t =(s,y,z)\right]
+  e^{-rt}(\gamma-1) K\\
& \le \gamma v(t,x)+ e^{-rt}(1-\gamma) K \\
& = \gamma \varphi(t,x)+ e^{-rt}(1-\gamma) K.
\end{align*}
As $v(t,x)=\varphi(t,s,z)>0$,
again by the homogeneity of $\alpha$ and $\beta$,
\begin{align*}
\varphi(t,\gamma s, \gamma z)&=v(t,x)\\ & = \sup_{\tau \in \vartheta_t}
\E\left[ \varphi(\tau,S_\tau,Z_\tau) \mid  X_t =(\gamma s, y, \gamma z)\right]
\\
&\le \gamma \varphi(t,s,z)+  e^{-rt}(1-\gamma) K\\
&= e^{-rt}\left(
-(\gamma \alpha(s,z)- \gamma \beta(s,z) -\gamma K)\right) +e^{-rt}(1-\gamma) K\\
&= e^{-rt}\left(-(
\alpha(\gamma s,\gamma z)- \beta(\gamma s,\gamma z) -  K)\right)\\
&= \varphi(t,\gamma s,\gamma z).
\end{align*}
Hence, $(\gamma s, y, \gamma z) \in \cS_t$.
\end{proof}

Let $f^*$
be as in  \eqref{e.fstar}. It is clear that if $\alpha(\bar{x}) > f^*(t,\Xi(\bar{x}))$,
then $\bar{x}\not \in \cS_t$.
Now, suppose $\alpha(\bar{x}) < f^*(t,\Xi(\bar{x}))$
for some $\bar{x}=(\bar{s}, \bar{y},\bar{z})$.  Then, 
by the above  definition, there exists 
$x=(s,y,z)\in \cS_t$ such that $ \alpha(x) \ge \alpha(\bar{x})$
and $\Xi(x)=\Xi(\bar{x})$.
By the definition of $\Xi$, we conclude that $y=\bar{y}$
and $(\bar{s},\bar{z})= (\gamma s,\gamma z)$
with $\gamma := \alpha(\bar{s},\bar{z})/\alpha(s,z)$.
Therefore, 
$(\gamma s, y, \gamma z) = \bar{x}$, and as $\gamma \le 1$,
by the above lemma
$\bar{x} \in \cS_t$.  Summarizing, we have shown that
$$
\{ \alpha < f^*\} \subset  \cS_t,\quad
\text{and}
\quad
\{ \alpha > f^*\}   \subset  \cX \setminus \cS_t.
$$
Now suppose that $\alpha(x)=f^*(t,\Xi(x))$.
Since $\alpha$ is always strictly positive, $f^*(t,\Xi(x))>0$.  Then, there is a sequence $x_n \in \cS_t$
with $\Xi(x_n)=\Xi(x)$ and $\alpha(x_n) \uparrow f^*(t,\Xi(x))=\alpha(x)$.
As $(\alpha, \Xi)$ is a homeomorphism, we 
conclude that $x_n$ converges to $x$.  Moreover,
$\cS_t$ is relatively closed in $\cX$, implying that
 $x \in \cS_t$.  Hence, $\{ \alpha = f^*\} \subset  \cS_t$.
Consequently, the triplet $(\alpha, \Xi, f^*)$ satisfies \eqref{e.st},
and the existence of a stopping boundary
Assumption \ref{asm:star} holds.

\section{Convergence of $\cR_\eps$}
\label{a.re}

The following limit result
justifies our choice of the reward function $\cR_\eps$. Further details are given in \cite{SonerTissot}. 

\begin{lemma}
\label{l.conv}
Suppose that $X_t$ has no atoms for all $t\in(0,T]$
 and $x_0=X_0$
is not on the boundary $f$, i.e., $\alpha(x_0) \neq f(0,\Xi(x_0))$.
Then, 
\begin{equation}
\label{e.elle}
\lim_{\eps \downarrow 0}\  \E[\cR_\eps(X,f)]=
\E[\varphi(\tau_f,X_{\tau_f})]= v(\tau_f).
\end{equation}
\end{lemma}
\begin{proof}
Under above assumptions, for every $t\in \cT^\circ$,
$ d(t,X_t;f)\neq0$ with probability one.
Then as $\eps$ tends to zero,
$p_t(X_t,f)$ converges to
one if $d(t,X_t;f)>0$ and
respectively, to zero if $d(t,X_t;f)<0$.
Consequently, $b_t(X,f)$ converges to
one for all $t\le \tau_f$ and
respectively, to zero for all $t>\tau_f$.
These limit statements directly imply that
$$
\lim_{\eps \downarrow 0}\ \cR_\eps(X,f)
= \varphi(\tau_f,X_{\tau_f}),\qquad
\text{a.s.}
$$
Moreover, as $b_t, p_t \in [0,1]$, Assumption \ref{a.growth}
implies that
$$
|\cR_\eps(X,f)| \le \sum_{t \in \cT} \varphi(t,X_t)
\le C \sum_{t \in \cT} [1+|X_t|^a]=:\Phi.
$$
Since $\cT$ is finite, by Assumption \ref{a.growth}, $\Phi$ is integrable.
Then, the claimed limit \eqref{e.elle} follows from 
the dominated convergence theorem.
\end{proof}

\section{Algorithm  Details}
\label{a.algo}

The following is a brief outline of the training
part of our algorithm.
\vspace{10pt}

{\underline{\bf{Stopping Boundary Training:}}}

\begin{enumerate}

\item \emph{Initialize} $\theta_{0} \in \Theta$
\item For $m = 0,\ldots, M-1$
    \begin{itemize}
        \item \emph{Simulate} trajectories $(X_{t_k}^j)_{k=0}^n$, $\ j=1,\ldots,B$.
        \item For $k=0,\ldots,n-1$, $j=1,\ldots$,
        \emph{compute} the following quantities:
\begin{align*}
-&\text{ signed distances} \quad &d_{k,m,j} &= 
\eta\ \left(g(t_k,\Xi(X^j_{t_k}); \theta^{m})-\alpha(X^j_{t_k})\right) \\
-&\text{ stopping probabilities}\quad &p_{k,m,j} &
= \left(\frac{\eps- d_{k,m,j}}{2\eps}\right)^+\wedge 1 
\quad \text{(} p_{n,m,j}=1 \text{)} \\
-&\text{ stopping budgets}\quad &b_{k+1,m,j} & = b_{k,m,j} (1- p_{k,m,j})
\quad \text{(}b_{0,m,j}=1 \text{)} \\
    -&\text{ reward function}\quad &R_\eps(\theta_m) &=  
    \frac{1}{B}\sum_{k=1}^B  \sum_{k=0}^n 
    p_{k,m,j}  b_{k,m,j} \varphi(t_k,X_{t_k}^j)
    \end{align*}
\item \emph{Update:} $\theta_{m+1} = \theta_{m} + \zeta_m \nabla_\theta R_\eps(\theta_{m})$
    \end{itemize}
\item \emph{Return} $\theta_{M}$ \qed
\end{enumerate}
\vspace{10pt}

In the above algorithm, we 
fix the width of the fuzzy region $\eps>0$,
and the learning rate 
process $\zeta_m$ is taken from the  Adam optimizer \cite{Kingma}.
Also in all financial examples, we use important 
sampling as discussed in the subsection \ref{ss.important}.
That is, chosen a parameter $\lambda$, we modify the
dynamics of the state process $X$ and use it in
the simulations.  Then, we adjust the reward function 
as in \eqref{e.lambda}.  The choice of the tuning
parameters $\eps$ and $\lambda$ are discussed in the specific
examples.

\bibliographystyle{abbrvnat}
\bibliography{refFB.bib}

\end{document}